\documentclass[prd,aps,amsfonts,notitlepage,longbibliography,twocolumn,superscriptaddress]{revtex4-1} 

\usepackage{graphicx}
\usepackage{tabularx}
\usepackage{braket}
\usepackage{amssymb}
\usepackage{amsmath}
\usepackage{bm}
\usepackage{tikz}
\usetikzlibrary{quantikz}
\usepackage{extpfeil}
\usepackage[driverfallback=dvipdfm]{hyperref}
\hypersetup{colorlinks=true,breaklinks,urlcolor=blue,linkcolor=blue,citecolor=blue}
\usepackage{etoolbox}

\makeatletter
\AtBeginEnvironment{thebibliography}{%
  \let\addcontentsline\@gobblethree
}
\makeatother

\newcommand{\comment}[1]{}
\usepackage{xcolor}
\definecolor{ginger}{rgb}{0.69, 0.4, 0.0}

\newcommand{\lr}[1]{\left( #1\right)}
\newcommand{\mlr}[1]{\left[ #1\right]}

\newcommand{\alr}[1]{\left\langle #1\right\rangle}
\newcommand{\norm}[1]{\left\lVert#1\right\rVert}
\newcommand{\abs}[1]{\left\lvert#1\right\rvert}

\newcommand{\ii}{\mathrm{i}}
\newcommand{\ee}{\mathrm{e}}
\newcommand{\dd}{\mathrm{d}}
\newcommand{\tr}[1]{\mathrm{tr}\lr{#1}}
\newcommand{\trace}{\mathrm{Tr}}

\newcommand{\where}{\quad {\rm where}\quad}
\newcommand{\order}{\mathrm{O}}

\newcommand{\floor}[1]{\left\lfloor#1\right\rfloor}

\newcommand{\poly}[1]{\mathrm{poly}\lr{#1} }

\newcommand{\OO}{\mathcal{O}}
\newcommand{\cU}{\mathcal{U}}

\newcommand{\expval}[2]{\left\langle{#2}\right\rangle_{#1}}
\newcommand{\mytag}[1]{\tag*{\llap{$\lr{#1}$}}}

\renewcommand{\thesection}{\arabic{section}}

\makeatletter
\renewcommand{\p@subsection}{}
\renewcommand{\p@subsubsection}{}
\makeatother

\usepackage{amsthm}
\newtheorem{thm}{Theorem}
\newtheorem{cor}{Corollary}
\newtheorem{lem}{Lemma}

\makeatletter
\newcommand{\normalfootnote}[1]{%
  \begingroup
    \let\ltx@footmark\ltx@footmark@latex
    \let\ltx@foottext\ltx@foottext@latex
    \footnote{#1}%
  \endgroup
}
\makeatother

\begin{document} 

\title{Digital Quantum Algorithms for Generating and Utilizing Spin Squeezed States}
\author{Mingru Yang}\email{mingruy@uci.edu}
\affiliation{Max-Planck-Institut f{\"u}r Quantenoptik, Hans-Kopfermann-Straße 1, D-85748 Garching, Germany}
\affiliation{Munich Center for Quantum Science and Technology (MCQST), Schellingstr. 4, D-80799 M{\"u}nchen, Germany}
\affiliation{Institute for Theoretical Physics, University of Cologne, Z{\"u}lpicher Straße 77, D-50937 K{\"o}ln, Germany}

\author{Ruby Wei}\email{ruby.wei@colorado.edu}
\affiliation{JILA, NIST and University of Colorado, Boulder, Colorado 80309, USA}
\affiliation{Department of Physics, University of Colorado, Boulder, Colorado 80309, USA}

\author{Chao Yin}\email{chaoyin@stanford.edu}
\affiliation{Department of Physics, Stanford University, Stanford, California 94305, USA}

\date{\today}

\begin{abstract}
Spin squeezed states (SSSs) are conventionally viewed as \emph{analog} resources for quantum metrology. Here we develop algorithms to efficiently generate and exploit SSSs on a \emph{digital} quantum computer. We introduce an adaptive local-circuit protocol that prepares SSSs with squeezing parameter $\xi$ in depth $\mathrm{O}\left(\log{1/\xi}\right)$, yielding an exponential improvement over non-adaptive approaches. We prove that this depth is optimal for permutation symmetric SSSs by deriving a matching lower bound on entanglement entropy. We further leverage the squeezing paradigm to deterministically prepare the Dicke states using asymptotically fewer resources than prior methods. Remarkably, classical simulations of our preparation protocols under realistic noise models predicts a metrological gain of 4.2 dB beyond the standard quantum limit with 55 superconducting qubits in total including ancillae. Our work paves the way for performing quantum metrology on near-term digital quantum devices.
\end{abstract}

\maketitle

\emph{Introduction.}--- Quantum metrology~\cite{PhysRevD.23.1693,PhysRevLett.72.3439,metro_rev04,PhysRevLett.96.010401,metro_QIrev14,metro_rmp17}  is one of the central applications of quantum technologies, which takes advantage of entanglement to achieve sensing precision beyond unentangled probes. For magnetic-field sensing, spin-squeezed states (SSSs)~\cite{squeeze92,squeeze93,MA201189} constitute a particularly important class of resource states. In addition to their potential to approach the Heisenberg limit, SSSs can exhibit optimal robustness against noise in metrology~\cite{squeeze_noise01,squeeze_noise15,sss_noise21,sss_noise24}, in sharp contrast to other Heisenberg-limited states such as the Greenberger-Horne-Zeilinger (GHZ) state~\cite{GHZ89}, whose metrological advantage is extremely fragile to decoherence. Beyond metrology, spin squeezing can also be utilized to accelerate quantum information processing~\cite{PhysRevLett.134.130604,all2all_compute25,all2all_GHZ26}. As an interesting family of many-body states, however, the \emph{entanglement structure and complexity of SSSs} from an information perspective remain partially understood.
Although their enhanced quantum Fisher information (QFI)~\cite{QFI67,QFI69} captures their metrological power, QFI alone cannot distinguish SSSs from other states like GHZ. Strong spin squeezing certifies two-particle entanglement~\cite{sss_entan05,sss_entan06}, but such criteria probe only few-body reduced states.
This urges a systematic understanding of resources needed to prepare SSSs, and what other tasks they can be useful for. 

For preparing SSSs, existing approaches predominantly employ collective analog dynamics that is qubit-permutation symmetric, where the problem reduces to squeezing a semiclassical large spin. One route among them is to evolve a coherent spin state (CSS)~\cite{CSS}, i.e. a product state, under a nonlinear all-to-all Hamiltonian~\cite{squeeze93,Sørensen2001,Quantumphasemagnification,PhysRevA.92.023603,PRXQuantum.4.020314,PhysRevX.11.041045,Luo2025}, while the others couple qubits to global boson modes to engineer collective dissipation~\cite{PhysRevLett.110.120402,adaptive_Clerk25,PhysRevX.12.011015} or perform continuous quantum nondemolition (QND) measurements~\cite{QND,RevModPhys.82.1041} of the collective angular momentum~\cite{PhysRevLett.85.1594,PhysRevLett.104.073604,PhysRevLett.116.093602,Hosten2016,6v93-whwq,PhysRevResearch.6.L032037}. However, these collective schemes rely on high qubit connectivity and are usually restricted to cavity-QED platforms \cite{cavity_rev13}.
Although the Hamiltonian-based approach has been extended to spatially decaying interactions~\cite{squeeze_powerlaw16,PhysRevLett.123.260505,PhysRevLett.125.223401,PhysRevLett.131.150401,Block2024,squeeze_adiab22,squeeze_powerlaw22,squeeze_powerlaw23,squeeze_2ddipo24,Franke2023,Eckner2023,Bornet2023,Miller2024,jqvz-kpwg,htbh-w6rx,wdqt-tpwz,Wu2025,squeeze_disorderdipo25,squeeze_powerlaw_Monika25}, these protocols require precise control of particular interaction ranges, or yield squeezing with scaling weaker than the Heisenberg limit.
Moreover, the above approaches in practice often suffer from noise that is intrinsic to the coupling that generates squeezing, which makes it hard to systematically detect and correct errors in such analog process.

In this Letter, we propose digital quantum algorithms as a complementary route to prepare SSSs, where better programmability and modularity naturally compatible with error correction provide a potential path toward fault-tolerant scalable quantum-enhanced sensing. 
In particular, we introduce an efficient adaptive protocol that prepares SSSs in optimal local circuit depth, which contributes another canonical example of measurement-assisted state preparation~\cite{PhysRevLett.76.722,PhysRevA.53.2046,PhysRevA.54.3824,PhysRevLett.86.5188,PhysRevA.68.022312,PhysRevA.71.062313,adaptive_phase21,adaptive_nonAbel22,adaptive_entanbound22,adaptive_measureSPT24,adaptive_Bravyi22,adaptive_shortroute23,adaptive_Nishimori23,Chen2025,adaptive_AKLT23,adaptive_Quantinuum24,adaptive_IBM24,adaptive_W_Buhrman24,Dicke_circuit24,MPS_log_loglog24,adaptive_tensor_Sahay24,adaptive_tensor_Zhang24,adaptive_MPS_Girvin24,adaptive_MPS_Stephen25,adaptive_1dsingleround25,adaptive_Slater_Bethe25,adaptive_variational_learning25,adaptive_variational25,PRXQuantum.4.020339,adaptive_solvableAnyon25,adaptive_AKLT_highd26,adaptive_Dicke_constdepth26,adaptive_push_defect26}, where intermediate measurements and classical feedforward can substantially reduce the time required to generate highly entangled states.
This circuit-based approach also enables us to uncover new information properties and applications of SSSs. To establish the optimality of our protocol, we prove that permutation-symmetric SSSs necessarily possess asymptotically large bipartite entanglement entropy, revealing a global entanglement property not captured by previous measures.
Furthermore, we use our squeezing circuit to prepare the Dicke states~\cite{dicke} with fewer resources than the previous proposals. Finally, we discuss methods to improve robustness of our protocols under noise and numerically benchmark their performance in near-term quantum hardware.

\emph{Setup.}---
$N$ qubits are in a SSS if the global polarization along some direction is extensive while the polarization variance along the perpendicular direction is suppressed compared to a CSS. More precisely, up to single-qubit rotations we can choose the two directions to be along $X=\sum_i X_i$ and $Z=\sum_i Z_i$, with $X_i,Y_i,Z_i$ Pauli matrices for each qubit, and define a pure state $\psi$ to be a $\xi$-SSS if~\footnote{We use notation $f=\order(g)$ ($f=\Omega(g)$) for $|f|\le c g$ ($|f|\ge c g$) for some constant $c$. We say $c$ is constant if it is independent of $N$ and $\xi$. $f=\Theta(g)$ means $f=\order(g)$ and $f=\Omega(g)$. We will also use $\poly{g}$ for $\order(g^c)$ for some constant $c$.} \begin{subequations}\label{eq:SSS_def}
    \begin{align}
        \expval{\psi}{X}&= \Theta(N), \label{eq:X=N} \\
        \expval{\psi}{\Delta Z^2} &\le N\xi^2. \label{eq:DZ2=xi}
    \end{align}
\end{subequations}
Here $\expval{\psi}{\cdot}$ denotes the expectation value, $\expval{\psi}{\Delta \OO^2}:= \expval{\psi}{\OO^2} -\expval{\psi}{\OO}^2$ is the variance, and $\xi<1$ is the squeezing parameter defined by Kitagawa and Ueda~\cite{squeeze93}. We are mainly interested in the asymptotics in $\xi\rightarrow 0$ for our rigorous results, which will also hold for the Wineland squeezing parameter $\xi_R$~\cite{squeeze92}. Because the $Z$ variance is suppressed by factor $\xi^2$ compared to CSS $\ket{+}^{\otimes N}$, when sensing $Y$-field in unitary $\ee^{\ii \omega Y}$, $\psi$ outperforms CSS in metrology precision $\lr{\delta\omega}_\psi=\order(\xi)\cdot\lr{\delta\omega}_{\rm CSS}$~\cite{squeeze92}. This implies a large $\mathrm{QFI}_\psi=\Omega(N/\xi^2)$ for sensing $Y$ field. By the Heisenberg uncertainty principle, \begin{equation}\label{eq:DeltaY2>}
    \expval{\psi}{\Delta Y^2}=\Omega( N/\xi^{2}),
\end{equation}
which imposes a fundamental limit $\xi=\Omega(N^{-\frac12})$ known as the Heisenberg limit. 
Most often, one is interested in the class of SSSs in the Dicke manifold (DM) of the $N$ qubits, i.e. the $(N+1)$-dimensional permutation-symmetric subspace. A basis for the DM is the set of Dicke states $\ket{D^N_M}$ ($M=0,1,\cdots,N$), which is equal superposition of all bitstrings with $M$ ones and $N-M$ zeros. Note that although the Dicke states have vanishing $Z$-variance, they are not SSSs due to vanishing polarization along perpendicular directions. 

One may first consider preparing SSSs by local unitary circuits. Note that we assume product initial states for all preparation protocols. Since the growth of QFI is bounded by spatial locality, any such circuit that prepares a $\xi$-SSS requires depth $D=\Omega(\xi^{-2/d})$ in $d$ spatial dimensions~\cite{QFI_growth_bound23}. We provide a simple proof for this in Supplementary Material (SM)~\cite{SM}, where we also show that any $\xi$-SSS in the DM with $\xi<1$ requires a depth growing with $N$. An example of SSS in the DM is $\ket{\Phi}=|\phi\rangle^{\otimes N}+|-\phi\rangle^{\otimes N}$ with $|\pm \phi\rangle_i:=\ee^{\mp\frac{\ii}{2}\phi Z_i}\ket{+}_i$ separated by angle $2\phi\sim1/\sqrt{N}$ on the Bloch sphere, as shown in Fig.~\ref{fig:binary}. One can verify that it is an $\xi$-SSS with $\xi=\Theta(1)<1$~\cite{SM}. Although locally indistinguishable from a CSS in the thermodynamic limit, $\ket{\Phi}$ is actually long-range entangled~\cite{SM}.

\begin{figure}
    \centering
    \includegraphics[width=\linewidth]{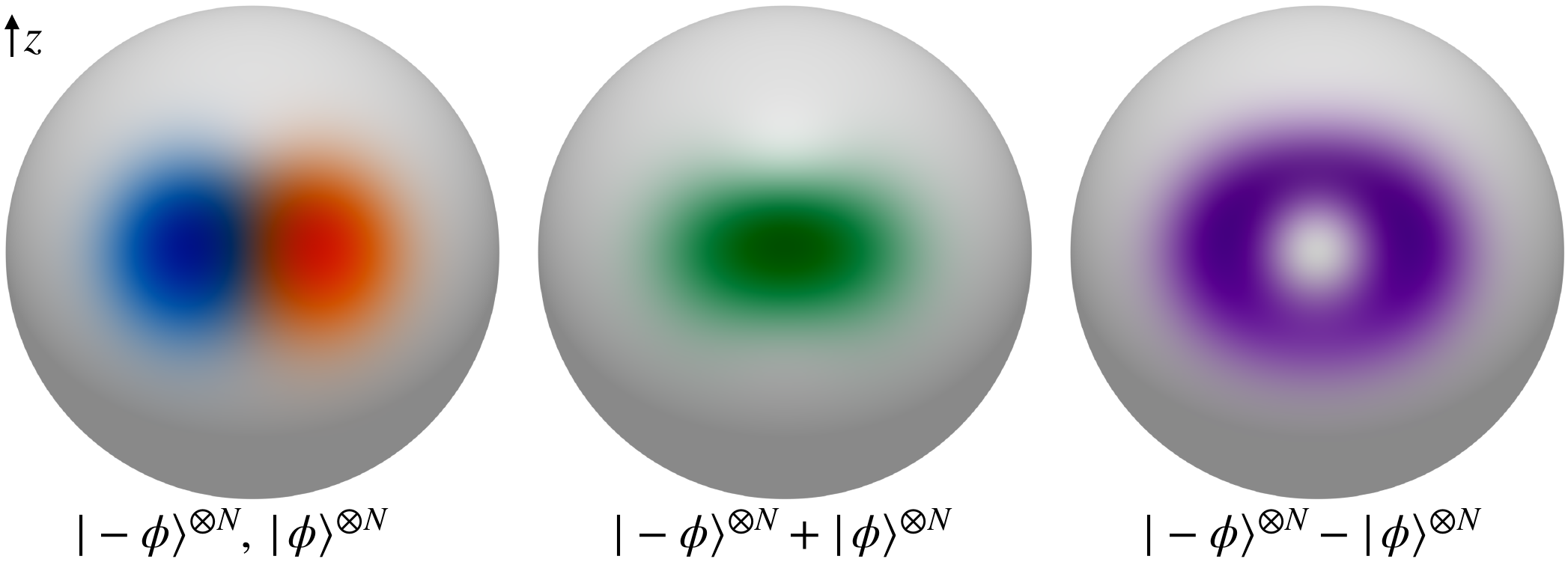}
    \caption{A sketch of a simple SSS $\ket{\Phi}$ (middle) in the DM, which is a superposition of two CSSs (left). Flipping the relative phase gives a state close to the $W$-state~\cite{3qubit_2way00,Dicke_circuit24} in a rotated basis (right), which is not a SSS.}
    \label{fig:binary}
\end{figure}

\emph{Generating SSSs using adaptive circuits.}---
Fig.~\ref{fig:binary} suggests a natural measurement-based approach to prepare the example $\ket{\Phi}$: one just needs to rotate the CSS $|+\rangle^{\otimes N}$ of $N$ data qubits controlled by an ancilla in $|+\rangle_a$, which gives $\ket{0}_a\otimes \ket{\phi}^{\otimes N}+ \ket{1}_a\otimes \ket{-\phi}^{\otimes N}=U_a \ket{+}^{\otimes N+1}$ with
\begin{equation}
    U_a = \ket{0}_a\bra{0}\otimes R_z(\phi_a) + \ket{1}_a\bra{1}\otimes R_z(-\phi_a),
\label{eq:controlU}
\end{equation}
where $R_z(\phi_a)=\ee^{-\frac{\mathrm{i}}{2}\phi_aZ}$ with $\phi_a=\phi$ here, and then project the ancilla to $\ket{+}_a$ by measurement which has constant success probability. Although this controlled rotation is a nonlocal operation, it can be done in $\order(1)$-depth local circuits using $\Theta(N)$ additional ancilla qubits together with local measurements and nonlocal classical feedforward~\cite{Dicke_circuit24}, e.g. by encoding the ancilla $a$ into a global GHZ state~\cite{adaptive_phase21}. Equivalently, the controlled rotation $U_a$ can be regarded as applying a global fanout gate (i.e. a product of \textsf{CNOT}s controlled by ancilla $a$) on state $R_z(\phi_a)|+\rangle^{\otimes}$; see Fig.~\ref{fig:recursivebinary} for a more general case utilized later.

Generalizing this idea, we present a 1d adaptive protocol that efficiently prepares $\xi$-SSSs for almost all $\xi$:
\begin{thm}\label{thm:prepare_sss}
    For any $\xi\ge c_1 N^{-\frac12}\log^6 N$ where $c_1$ is a numerical constant, an $N$-qubit $\xi$-SSS in the DM can be generated in an 1d adaptive circuit of depth $D=\order(\log \frac{1}{\xi})$ and width $2N+\order(\log \frac{1}{\xi})$ with constant success probability $P_{\rm succes}= 1-\order(1/\log(\xi^{-1}))$.
\end{thm}
Up to a technical $\mathrm{polylog} N$ factor, Theorem~\ref{thm:prepare_sss} covers the whole range of $\xi$ up to the Heisenberg limit. The depth increases only logarithmically with $1/\xi$, \emph{exponentially} smaller than unitary circuits discussed above.
In particular, the depth does not need to scale with the system size to produce $\xi\ll 1$. For metrological purpose, although the GHZ state has maximal QFI and can be prepared in $\order(1)$ depth~\cite{adaptive_phase21}, it is extremely fragile to noise in that a tiny noise strength $p\gtrsim 1/N$ per qubit is already enough to destroy its quantumness~\cite{metro_GHZ_noise97}. In contrast, our ideal protocol produces metrological resource states that are optimally robust~\cite{sss_noise24}, where it can tolerate system-size independent noise $p\lesssim \xi^2$ (see SM). Our preparation protocol, however, is not as robust to noise as the final state due to the nonlocal control step \eqref{eq:controlU}. Nevertheless, our protocol could be useful when noise during preparation is much smaller than noise during sensing, and its noise robustness can be improved, as we will see.

We sketch the proof of Theorem~\ref{thm:prepare_sss} here, leaving calculation details in SM. 
Our protocol is illustrated in Fig.~\ref{fig:sss_circuit}, which merges the idea of block encoding~\cite{PhysRevLett.114.090502} of linear combination of unitaries (LCU)~\cite{LCU2012} and  phase estimation~\cite{Kitaev_phaseEstimation}, resembling adaptive circuits that prepare Dicke states \cite{Dicke_phaseEstimate21,Dicke_circuit24}. Each data qubit that supports the final SSS is neighboring one ancilla qubit labeled by $b$ (not shown in Fig.~\ref{fig:sss_circuit}), and there are extra $\ell=\order(\log\frac{1}{\xi})$ ancillae labeled by $A=\{a:a=1,\cdots,\ell\}$ at one end of the system. The data qubits are initialized to $\ket{+}$ while the ancillae are $\ket{0}$. The protocol consists of 4 stages: (\emph{1}) A ``Gaussian'' state $\ket{\alpha}=\sum_x \alpha_x\ket{x}$ is prepared on $A$, where $x=-2^{\ell-1},\cdots,2^{\ell-1}-1$ labels the (shifted) computational basis and $\alpha_x=\mathcal{N}\ee^{-\pi{\sigma^2}x^2}$, where $\mathcal{N}$ is normalization and $\sigma$ is chosen shortly. (\emph{2}) Each ancilla $a$ controls rotation of data qubits by $U_a$ in Eq.~\eqref{eq:controlU} with $\phi_a=2^{a}\pi\phi$ chosen shortly, i.e. $\mathcal{U}=\prod_{a=1}^\ell U_a$.
In Fig.~\ref{fig:sss_circuit}, we have written $\cU$ in an equivalent form such that each ancilla just applies a global fanout on the data qubits, where the $\phi_a$ rotations become on-site and there are two extra layers of $\textsf{GRAY}=\prod_{a=1}^{\ell-1}\textsf{CNOT}_{a+1,a}$ implementing the Gray code~\cite{gray1953}. This form generalizes the equivalence shown in Fig.~\ref{fig:recursivebinary},  and will be useful when we discuss noise robustness. Overall, the stage $\cU$ takes $\order(\ell)$ depth by routing and treating the $a$'s sequentially.
(\emph{3}) The inverse quantum Fourier transform (QFT) is performed on $A$, which also takes $\order(\ell)$ depth \cite{PhysRevA.76.052310}. (\emph{4}) $A$ is finally measured in the computational basis labeled by $m=-2^{\ell-1},\cdots,2^{\ell-1}-1$.

\begin{figure}
    \centering
\scalebox{0.63}{
\begin{quantikz}
\lstick{$|0\rangle_{a_l}$}&\gate[4,nwires=2,style={minimum height=42mm,inner xsep=1pt,inner ysep=0pt}]{\smash{
    \rotatebox[origin=c]{90}{\textsf{PREP}}}}&\gate[4,nwires=2,style={minimum height=42mm,inner xsep=1pt,inner ysep=0pt}]{\smash{
    \rotatebox[origin=c]{90}{\textsf{GRAY}}}}\gategroup[5,steps=8,style={dashed,rounded corners,fill=blue!5, inner xsep=2pt, inner ysep=8pt},background,label style={label position=below,anchor=north,yshift=-0.2cm}]{$\mathcal{U}$}       &\qw&\qw       &\qw&\qw\ \ldots\ &\qw&\ctrl{4}  &\gate[4,nwires=2,style={minimum height=42mm,inner xsep=1pt,inner ysep=0pt}]{\smash{
    \rotatebox[origin=c]{90}{\textsf{GRAY}$^{-1}$}}}&\gate[4,nwires=2,style={minimum height=42mm,inner ysep=0pt}]{\smash{
    \rotatebox[origin=c]{90}{\textsf{QFT}$^{-1}$}}}&\meter{} \\
\lstick{\vdots}           &                                &&          &&          &  \ \vdots\   &&          &&&\vdots&\\
\lstick{$|0\rangle_{a_2}$}&                                &&\qw       &\qw&\ctrl{2}  &\qw\ \ldots\ &\qw&\qw       &&&\meter{}\\
\lstick{$|0\rangle_{a_1}$}&                                &&\ctrl{1}  &\qw&\qw       &\qw\ \ldots\ &\qw&\qw       &&&\meter{}\\
\lstick{$|+\rangle^{\otimes N}$}   &\qwbundle{N}                    &\gate[1]{\mbox{\fontsize{6pt}{6pt}$R_z(\phi_1)$}}&\targ{}&\gate[1]{\mbox{\fontsize{6pt}{6pt}$R_z(\phi_2)$}}&\targ{}&\qw\ \ldots\ &\gate[1]{\mbox{\fontsize{6pt}{6pt}$R_z(\phi_\ell)$}}&\targ{}&\qw&\qw&\qw
\end{quantikz}
}
    \caption{An adaptive protocol that prepares a $\xi$-SSS from $|+\rangle^{\otimes N}$.
    }
    \label{fig:sss_circuit}
\end{figure}
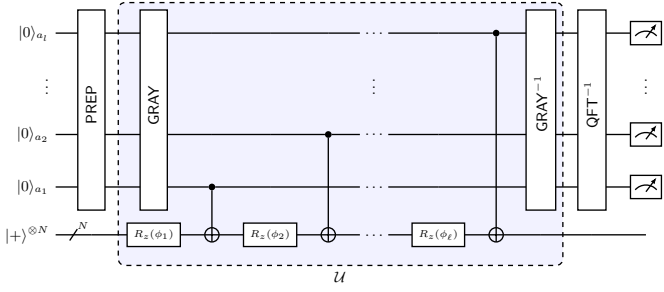

At those stages, the state on $A$ $+$ data qubits is
\begin{align}\label{eq:stages}
    &\hspace{-2em}\sum_x \alpha_x \ket{x} \ket{+}^{\otimes N} \overset{(\emph{2})}{\longmapsto} \sum_x \alpha_x \ee^{2\pi\ii\phi x Z}\ket{x} \ket{+}^{\otimes N} \nonumber\\
    & \overset{(\emph{3})}{\longmapsto} \sum_x \alpha_x\ee^{2\pi\ii\phi x Z} \lr{\frac{1}{2^{\frac{\ell}{2}}} \sum_m \ee^{-\ii \frac{2\pi}{2^\ell}mx}\ket{m}} \ket{+}^{\otimes N} \nonumber\\
    &= \frac{1}{2^{\frac{\ell}{2}}}\sum_m \ket{m} \otimes \lr{\sum_x \alpha_x \ee^{2\pi\ii x\lr{\phi Z-2^{-\ell}m}}} \ket{+}^{\otimes N},
\end{align}
where we have ignored an irrelevant global phase, and the QFT differs from the standard one by two extra layers of single-qubit gates. The final measurement projecting to a given $m$ effectively performs a \emph{coarse-grained} measurement of the collective $Z$ on the initial $\ket{+}^{\otimes N}$, yielding an unnormalized SSS $\ket{\psi_m}= f_m(Z)\ket{+}^{\otimes N}$. It can also be viewed as filtering a CSS, with the Gaussian filter operator $f_m(Z)$ decomposed into a LCU by the Fourier transformation~\cite{filter19}. This is a crucial difference comparing to the Dicke-state protocol \cite{Dicke_circuit24} where observable $Z$ is fully resolved, which destroys polarization in perpendicular directions and is equivalent to filtering with the Dirichlet kernel. More precisely, our filter operator
\begin{align}\label{eq:fm=}
    f_m(Z)
    &\approx \mathcal{N} \sum_{x\in \mathbb{Z}} \ee^{-\pi\sigma^2x^2} \ee^{2\pi\ii x \lr{\phi Z-2^{-\ell}m}} \nonumber\\
    &= \mathcal{N} \sigma^{-1} \sum_{q\in\mathbb{Z}} \exp\mlr{-\frac{\pi}{\sigma^2}\lr{\phi Z-2^{-\ell}m+q}^2}
\end{align}
decays exponentially fast but does not strictly vanish, where we have chosen $\sigma=2^{-\ell}\ell$ so that relaxing the sum over $x$ to all integers in the first line incurs a small error $\exp\mlr{-\Omega(\ell^2)}$. The second line of \eqref{eq:fm=} follows from the Poisson summation formula. $f_m(Z)$ is thus supported in periodic narrow windows of width $\Delta Z\sim \sigma/\phi$ and period $\phi Z\rightarrow \phi Z+1$, whose centers shift with $m$. Since $\ket{+}^{\otimes N}$ is mainly supported in $Z\in [-\sqrt{N},\sqrt{N}]$, we choose $\phi\sim N^{-1/2}$ so that we can focus on one period, making $\ket{\psi_m}$ satisfy \eqref{eq:DZ2=xi} with $\xi\sim 2^{-\ell}$ for almost all measurement outcomes. Here the Gaussian initial state $\ket{\alpha}$ is crucial to guarantee a fast decaying tail outside the $Z$ window. Furthermore, $\psi_m$ is $\xi$-SSS because it also satisfies \eqref{eq:X=N}. To see this, observe that the controlled rotation angle is at most $\sim 2^{\ell}\phi=\order(\phi/\xi)\ll 1$ for any $\xi\gg N^{-1/2}$ considered. Therefore, starting from all-plus, the data qubits remain to have a $\Theta(N)$ polarization in $X$ throughout the protocol, which turns out to hold for almost every measurement outcome.

To establish Theorem \ref{thm:prepare_sss}, it remains to show that $\ket{\alpha}$ can be prepared in $\order(\ell)$ depth.
\comment{
Again using the Poisson summation formula \eqref{eq:fm=}, \begin{equation}\label{eq:x2=y2}
    \ee^{-\pi{\sigma^2}x^2} \approx \sum_{q\in \mathbb{Z}}\ee^{-\pi \widetilde{\sigma}^2(\widetilde{x}+q)^2} = \frac{1}{\widetilde{\sigma}} \sum_{y\in \mathbb{Z}} \ee^{-\frac{\pi}{\widetilde{\sigma}^2}y^2} \ee^{2\pi\ii\widetilde{x}y}
\end{equation}
where $\widetilde{\sigma}=\ell,\widetilde{x}=2^{-\ell}\ell^2 x$ so that the approximation error is again small $\exp\mlr{-\Omega(\ell^{4})}$. With similar amount of error, we further truncate the sum to $|y|\le \ell^{2}$ in \eqref{eq:x2=y2}. $\ket{\alpha}$ is then approximately QFT (doable in $\order(\ell)$ depth) acting on Gaussian state $\sum_{|y|\le \ell^2} \ee^{-\frac{\pi}{\widetilde{\sigma}^2}y^2} \ket{y}$ that is effectively supported on $\order(\log \ell)$ qubits, which can be prepared in $\mathrm{polylog}(\ell)\ll \ell$ depth \cite{Gaussian_Kitaev08}. 
}
We prove the following general result on preparing Gaussian states in SM, which completes the proof sketch for Theorem \ref{thm:prepare_sss}. We also give a MPS representation of the Gaussian state and prove it has small entanglement in SM. 

\begin{lem}
\label{lem:prepGaussian}
The $\ell$-qubit state $|\alpha\rangle=\sum_x\alpha_x|x\rangle$, where $\alpha_x=\mathcal{N}e^{-\pi\sigma^2x^2}$ with $\sigma=\Omega(2^{-\ell}\ell)$ are Gaussian amplitudes, can be prepared in 1d local unitary circuit in depth $O(\ell)$ without ancillae, up to state error $\epsilon=e^{-\Omega(\ell^2)}$.
\end{lem}

\emph{Entanglement entropy.}---
To show optimality of the above protocol, we look into entanglement structure of SSSs, quantified by the smooth max-entropy \cite{Renner_thesis05}  \begin{equation}
    H^{\epsilon}_{\rm max}(\rho) := \min_{\widetilde{\rho}\ge 0:\, \norm{\widetilde{\rho}-\rho}_1\le \epsilon} H_{\rm max}(\widetilde{\rho}),
\end{equation}
where $\rho$ is a normalized density matrix, and $H_{\rm max}(\widetilde{\rho}) := \log \mathrm{rank}(\widetilde{\rho})$ is the entanglement rank or the Rényi-0 entropy~\footnote{Note that we ignore the constraint $\tr{\protect\widetilde{\rho}}\le\tr{\rho}$ in the minimization \cite{Renner_thesis05}; \eqref{eq:maxentropy>} still holds if we impose this additional constraint.}.
Although the von Neumann entanglement entropy $H(\rho)=-\tr{\rho \log \rho}$ is more widely used to quantify bipartite entanglement, the smooth max-entropy is a one-shot version that bounds the resource to prepare a single approximate copy of the state. Our second main result is:

\begin{thm}\label{thm:bound}
    Any $\xi$-SSS in the DM has bipartite entanglement \begin{equation}\label{eq:maxentropy>}
        H^{\epsilon}_{\rm max}(\rho_L) =\Omega\lr{ \log \frac{1}{\xi}},
    \end{equation}
for any constant $\epsilon\in (0,1)$, where $\rho_L$ is the reduced density matrix of $N/2$ qubits. As a result, preparing any such $\xi$-SSS, even approximately, requires a 1d adaptive circuit of depth $\Omega(\log(1/\xi))$.
\end{thm}

Here we assume $N$ is even for simplicity.
This provides a many-body entanglement measure of SSSs beyond QFI, which differentiates SSSs from the GHZ state which has $\order(1)$ bipartite entanglement, suggesting that the optimally robust metrological power of SSSs~\cite{sss_noise24} may fundamentally require substantial resources. Theorem~\ref{thm:bound} requires the state to be permutation symmetric, because otherwise it can be a tensor product of two SSSs of $N/2$ qubits that has zero entanglement across the cut. Nevertheless, it is plausible that for a general SSS, there exists a bipartition of the system with $\Omega(\log \frac{1}{\xi})$ entanglement entropy; we leave this general case for future work. The entanglement bound shows that our protocol in Theorem~\ref{thm:prepare_sss} is optimal for generating SSSs in the DM, because 1d adaptive circuits can increase entanglement entropy at most linearly in circuit depth~\cite{adaptive_entanbound22}.

We prove Theorem~\ref{thm:bound} in SM, with the intuition being the uncertainty principle: The standard uncertainty relation~\eqref{eq:DeltaY2>} is not sufficient because it does not exclude GHZ-like states with $\order(1)$ entanglement. We need a generalized uncertainty relation called \emph{support uncertainty principles} or \emph{entropic uncertainty}~\cite{uncertain_support89,uncertain_sparse02,uncertain_frame03,uncertain_refine13,uncertain_rmp17,uncertain_product24}; see~\cite{uncertain_rev14} for a survey. Its simplest form is as follows: Any function $f_x:x=0,\cdots,n-1$ with its discrete Fourier transform $\hat{f}_k$ satisfy $n_x\cdot n_k \ge n$, where $n_x$ ($n_k$) is the number of nonzero $f_x$ ($\hat{f}_k$)~\cite{uncertain_support89}. In other words, $f$ and $\hat{f}$ cannot be both concentrated in sets of small cardinality. For a $\xi$-SSS $\ket{\psi}$ in the DM polarized along $X$, observables $\frac{Y}{\sqrt{N}},\frac{Z}{\sqrt{N}}$ effectively form Fourier conjugates $x,k$, so \eqref{eq:DZ2=xi} requires the ``support cardinality'' of $Y$ to be large $n_{Y/\sqrt{N}}\gtrsim \xi^{-1}$, excluding the GHZ example. Since the DM is spanned by CSSs $\ket{\omega}^{\otimes N}$ where $\ket{\omega}$ is a state on the single-qubit Bloch sphere with solid angle $\omega$, \begin{equation}\label{eq:psi=omega}
    \ket{\psi}\approx \sum_\omega \psi_\omega \ket{\omega}=\sum_\omega \psi_\omega \ket{\omega_L}\otimes \ket{\omega_R},
\end{equation}
where the weight $|\psi_\omega|^2$ is spread out in $\sim n_{Y/\sqrt{N}}\gtrsim \xi^{-1}$ different values of $\omega$. Note that \eqref{eq:psi=omega} only includes $\omega$ that are $\gtrsim N^{-1/2}$ apart so that each pair $\ket{\omega}^{\otimes N},\ket{\omega'}^{\otimes N}$ is almost orthogonal. We have also expanded $\ket{\omega}=\ket{\omega_L}\otimes \ket{\omega_R}$ to CSSs of the two $\frac{N}{2}$-qubit subsystems $L,R$ in \eqref{eq:psi=omega}. Because $\{\ket{\omega_L}\}$ are also approximately orthogonal, \eqref{eq:psi=omega} leads to entanglement $\sim \log n_{Y/\sqrt{N}}\gtrsim \log(1/\xi)$. 

\emph{Preparing Dicke states via squeezing.}--- Dicke states $\ket{D^N_M}$ are another important family of states with wide applications in quantum technologies~\cite{Dicke_qcommuni09,metro_QIrev14,Dicke_permu_code14,Dicke_qsimu21,DQI25,Dicke_symmetrize26}. Ref.~\cite{Dicke_circuit24} provides a protocol to prepare the Dicke state for any $M$ in 1d adaptive circuits, which generally requires depth $\Theta(\log^2N)$ for deterministic preparation when combined with the adaptive rotation scheme in Ref.~\cite{Dicke_log3_24} to convert one Dicke state to another via measurements. However, this depth is larger than the lower bound $\Omega(\log N)$ given by the $\Theta(\log N)$ bipartite entanglement of Dicke states. Instead, generalizing the previous squeezing protocol, we are able to achieve this optimal scaling:

\begin{figure}
    \centering
\scalebox{0.5}{
\begin{quantikz}
\lstick{$|0\rangle_{a}$}&\qw&\gate[1]{H}\gategroup[2,steps=5,style={dashed,rounded
corners,fill=blue!5, inner
xsep=2pt},background,label style={label
position=below,anchor=north,yshift=-0.2cm}]{{Repeat $k$ rounds}}&\gate[2]{U_a}&\gate[1]{H}&\meter{}&\gate[1]{|0\rangle}&\qw\ \\
\lstick{$|\Psi\rangle$}&\qwbundle{N}&\qw&\qw&\qw&\qw&\qw&\qw\
\end{quantikz}
$\xlongequal{|\Psi\rangle=\prod_i X_i|\Psi\rangle}$
\begin{quantikz}
\lstick{$|0\rangle_{a}$}&\qw&\gate[1]{H}\gategroup[2,steps=5,style={dashed,rounded
corners,fill=blue!5, inner
xsep=2pt},background,label style={label
position=below,anchor=north,yshift=-0.2cm}]{{Repeat $k$ rounds}}&\ctrl{1}&\gate[1]{H}&\meter{}&\gate[1]{|0\rangle}&\qw\ \\
\lstick{$|\Psi\rangle$}&\qwbundle{N}&\gate[1]{R_z(\phi_a^{(r)})}&\targ&\qw&\qw&\qw&\qw&\qw\
\end{quantikz}
}
    \caption{An alternative squeezing protocol, where the angle at round-$r$ is $\phi_a^{(r)}=\frac{2}{q_{\rm f}}\frac{\sqrt{\langle\Delta Y^2\rangle^{(r-1)}}}{|\langle X\rangle^{(r-1)}|}$. Here superscript $(r-1)$ means quantities of the state after round-$(r-1)$, and $q_{\rm f}$ is a constant one can choose. The protocol prepares $\ket{\Phi}$ at $k=1$, and the Dicke state $\ket{D^N_{N/2}}$ at $k=\log_2N$ for a different angle $\phi_a^{(r)}=2^{r-1}\pi/N$ with success probability $\Theta(1/\sqrt{N})$ (notice this protocol is different from Theorem~\ref{thm:dicke}).
    }
    \label{fig:recursivebinary}
\end{figure}
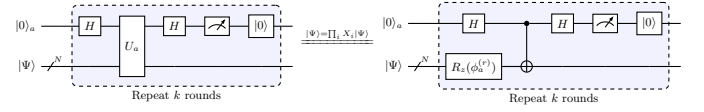

\begin{thm}\label{thm:dicke}
    The Dicke state $\ket{D^N_{N/2}}$ can be prepared in 1d adaptive circuits of depth $D=\order(\log N)$ and width $2N+\order(\log N)$ with $\Omega(1)$ success probability.
\end{thm}

Here we focus on the equator Dicke state $M=N/2$ that is the most challenging to prepare \cite{Dicke_log3_24,Dicke_symmetrize26}; we believe the result can be generalized to any Dicke state. Theorem \ref{thm:dicke} is proven in SM; the key observation is that coarse-grained measurements like \eqref{eq:fm=} are \emph{cheaper} than measurements with full resolution: $\order(1)$ depth is sufficient to squeeze $\alr{\Delta Z^2}$ by a constant factor, while measuring $Z$ as in Ref.~\cite{Dicke_circuit24} requires $\Omega(\log N)$ depth. Indeed, the squeezing protocol above only has depth $\order(\log N)$ for SSSs close to the Heisenberg limit, and the Dicke state can be viewed as the \emph{oversqueezed} limit. Nevertheless, we need to repeat the previous protocol in an recursive way with an adaptive rotation scheme following \cite{PhysRevLett.116.093602,Dicke_log3_24} in order to achieve the precise Dicke state, where each recursive stage does not squeeze the state too much such that squeezing in $Z$ direction in almost all measurement outcomes does not deteriorate much when rotated back to the equator.

Since each fanout in Fig.~\ref{fig:sss_circuit} can be realized by $\order(\log N)$-depth unitary circuits with all-to-all connectivity and $\order(N)$ ancillae, Theorem \ref{thm:dicke} also implies an efficient all-to-all unitary circuit to prepare Dicke states: 
\begin{cor}
    The Dicke state $\ket{D^N_{N/2}}$ can be prepared in all-to-all unitary circuits of depth $D=\order(\log^2 N)$ and width $2N+\order(\log N)$ with $\Omega(1)$ success probability.
\end{cor}
With $\order(N)$ ancillae, this protocol improves upon the previous best depth $\order(\log^3 N)$ obtainable from Ref.~\cite{Dicke_log3_24}, where the $\log N$ factor improvement again comes from replacing exact $Z$ measurements with coarse-grained ones. With $\Theta(N\log N)$ ancillae, $\order(\log N)$ depth is achievable~\cite{Dicke_symmetrize26} using a different approach of symmetrization, which however does not seem to generalize directly for constant space overhead. It is still an open question whether an $\order(\log N)$-depth unitary preparation circuit exists with $\order(N)$ ancillae.

\emph{Performance under noises.}--- Now we consider our adaptive-circuit squeezing protocol in realistic settings with noise, and we expect the Dicke state preparation to behave similarly. Our protocol has two appealing aspects, namely the depth is short and the final SSS has optimal noise robustness~\cite{sss_noise24}. Indeed, Pauli-$X$ errors on the data qubits do not propagate in our protocol and are equivalent to noise channels acting on the final state, so the final squeezing will remain $\order(\xi)$ up to noise strength $p=\order(\xi^2/\log\frac{1}{\xi})$ per qubit per layer during the protocol, if that is the only error source. However, the circuit can be much more sensitive to other noise sources, especially for the global fanout gates in Fig.~\ref{fig:sss_circuit}. In the 1d layout shown in Fig.~\ref{fig:noisyprotocol}(a), the fanout can be realized by initializing each edge ancilla qubit in $\ket{0}$, measuring $S_v=X_v\prod_{e\sim v}X_e$ for each vertex qubit $v$ and its neighboring edges $e$, measuring the edge qubits in $Z_e$, and then adaptively applying single-qubit gates based on those measurement results. The idea is that a final edge is $\ket{1}$ if and only if one of its two neighboring vertices is flipped by $X_v$, so an adaptive feedforward ensures all vertices are flipped at the same time. On the other hand, the global parity of all the $S_v$ measurements determines whether to apply a $Z_a$ on the reference qubit $a$. Nevertheless, the ideal implementation of the fanout above can be destroyed easily by a single measurement fault because the feedforward relies on the global parity of the measurement outcomes.

This motivates several modifications of our protocol to mitigate noise. First, each measurement can be repeated multiple rounds to correct errors or simply postselect the case where measurement results agree.
Second, since any one of the fanouts in Fig.~\ref{fig:sss_circuit} could fail, we consider the modified protocol in Fig.~\ref{fig:recursivebinary} with only one reference ancilla $a$ and one fanout in each round, where each round only aims to squeeze an already squeezed state further by a constant factor. This modified protocol generalizes the previous example $\ket{\Phi}$, and achieves squeezing by moderate amount of postselection of the reference readouts.
Third, we generalize the fanout implementation to 2d as shown in Fig.~\ref{fig:noisyprotocol}(a), which is fundamentally robust against the edge measurement errors. The idea is that the ideal edge measurements should satisfy $\prod Z_e=+1$ for each four edges around a plaquette, which can be utilized to correct the measurement errors like the 2d Ising model as a repetition code~\cite{qmemory02,PhysRevX.5.031043}. In fact, this 2d layout deals with all types of single-qubit errors on edges throughout the protocol, because an $X_e$ error flips the $Z_e$ readout which can be almost corrected, while the $Z_e$ errors do not flip the total parity of $S_v$ measurements.
Among all types of single-qubit errors, the above modifications leave $Z_v$ errors unaddressed, which could flip the $S_v$ measurements and make the feedforward operation ambiguous.
Nevertheless, we show that the above techniques already give rise to significant improvements in realistic settings.

\begin{figure}
    \centering
    \includegraphics[width=\linewidth]{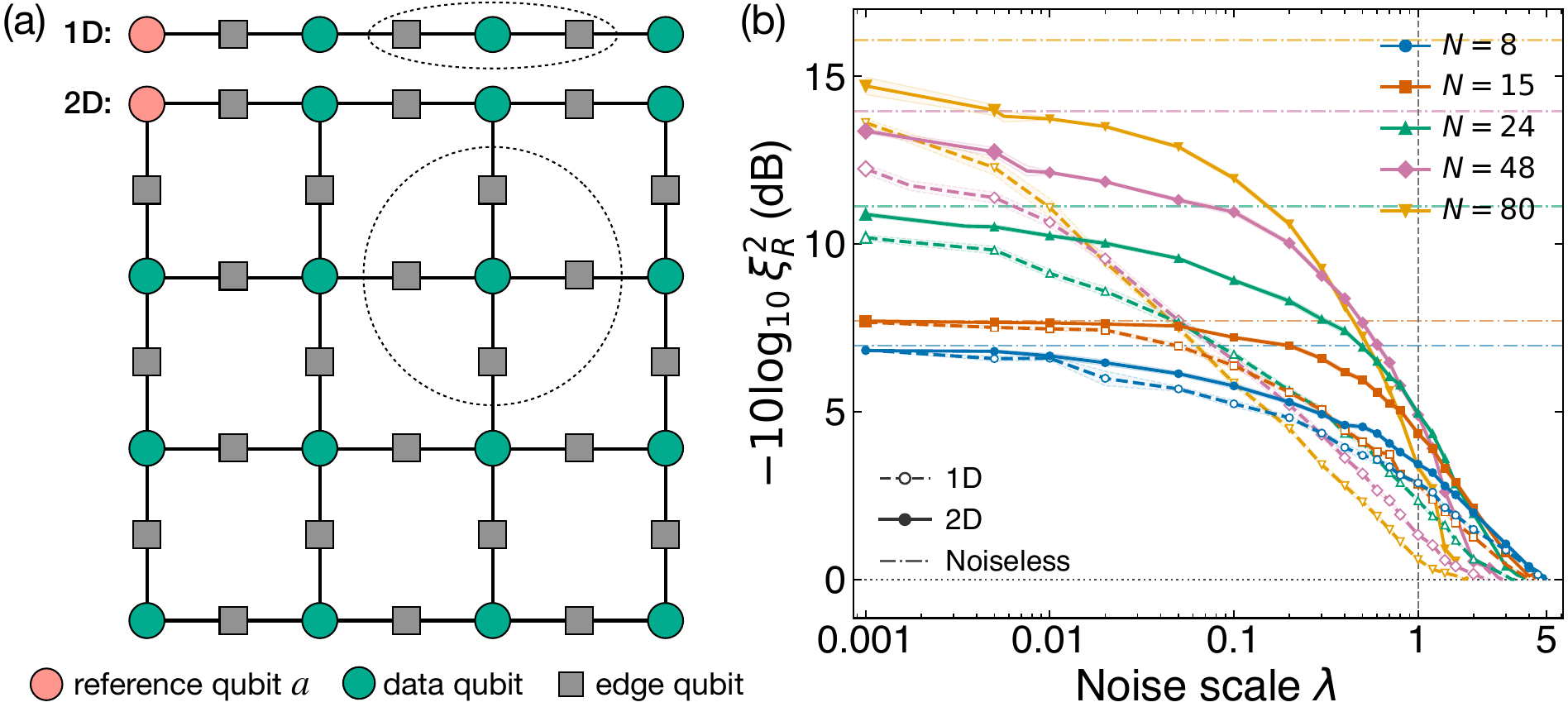}
    \caption{(a) Illustration of the Lieb lattice~\cite{PhysRevLett.62.1201} geometry of the 1d and 2d protocol.
    (b) Wineland spin squeezing parameter $\xi^2_R$ obtained from the 1d and 2d protocol versus the noise scale $\lambda$ for different system sizes $N$. We switch from the protocol in Fig.~\ref{fig:recursivebinary} (smaller marker) to that in Fig.~\ref{fig:sss_circuit} (larger marker) when tuning down $\lambda$. The noise level used in the simulation is set to be $\lambda$ times $(p_2,p_1,p_{\mathrm{readout}},p_{\mathrm{idle}},p_{\mathrm{init}})=(0.8,0.02,1,0.001,0.1)\%$ which are the typical stochastic incoherent raw noise rates for two- and one-qubit depolarizing errors, classical readout bit flips, idle Pauli errors, and initialization errors in the current superconducting qubit devices.}
    \label{fig:noisyprotocol}
\end{figure}

Under a local stochastic error model~\cite{SM}, Fig.~\ref{fig:noisyprotocol}(b) shows the metrological gain of the output state measured by $\xi_R^2=N\alr{\Delta Z^2}/\alr{X}^2$ for different system sizes and noise levels. The details of the simulation can be found in SM. In particular, we measure each $S_v$ twice to postselect, and use a classical minimum-weight-perfect-matching (MWPM) decoder~\cite{Edmonds_1965,qmemory02,lee2022decodingmeasurementpreparedquantumphases} to decode the edge readouts. We use both protocols in Fig.~\ref{fig:recursivebinary} and Fig.~\ref{fig:sss_circuit} with parameters $\ell,\phi,\sigma,k$ optimized, and show whichever has the smaller $\xi_R$. The results indicate that the protocol in Fig.~\ref{fig:recursivebinary} is more robust under larger noises than the protocol in Fig.~\ref{fig:sss_circuit}. A remarkable result is that at the current noise level of superconducting devices, our 2d protocol enables 4.2~dB metrological gain using only 15 data qubits (55 qubits in total including ancillae, $k=5$), which is comparable to recent experiments on analog cold atom platforms~\cite{Franke2023,Eckner2023,Bornet2023}. Although the postselection acceptance rate 0.25\% is not large, the experiment can be repeated utilizing the fast operation speed of superconducting circuits.
Furthermore, if the noise level improves to $\sim 1/2$ of the current values, we see already that $\xi_R$ keeps decreasing with increasing system size, suggesting scalability of our approach up to at least hundreds of qubits.

\emph{Conclusion and outlook.}--- We have developed efficient digital quantum algorithms for preparing SSSs and Dicke states, and shown that strong squeezing necessarily requires substantial bipartite entanglement. A natural next step is to demonstrate these algorithms on real quantum hardware, such as superconducting circuits. Moreover, our digital protocols offer a route to incorporate quantum error correction, which could potentially further improve the squeezing performance and advance quantum computational sensing~\cite{qcompusensing_rev26}. Generalizing the ideas of our protocols, it would also be interesting to see whether squeezing can serve as an algorithmic primitive for accelerating more general quantum information processing tasks.

\emph{Acknowledgements.}---
We thank Zhenning Liu, Yilun Yang, Georgios Styliaris, and Emanuel Knill for valuable discussions. C.Y. was supported by the Department of Energy under Quantum Pathfinder Grant DE-SC0024324, and is partially supported by the National Science Foundation under Award No.~2016245 and by the Stanford Q-FARM Bloch Postdoctoral Fellowship in Quantum Science and Engineering.
M.Y. was supported by the MCQST Distinguished Postdoc Fellowship funded by the Deutsche Forschungsgemeinschaft (DFG) under Germany’s Excellence Strategy EXC2111-390814868, and is supported in part by DFG under the Excellence Strategy—Cluster of Excellence Matter and Light for Quantum Computing (ML4Q) EXC2004/1–390534769.
R.W. acknowledges funding from
the U.S. Department of Energy, Quantum Systems Accelerator, NSF PFC grant No. PHYS 2317149, and the Natural Sciences and Engineering Research Council of Canada (NSERC) through a Canada Graduate Research Scholarship – Doctoral (CGRS D).

\bibliography{biblio}

\onecolumngrid

\newpage

\setcounter{equation}{0}
\setcounter{figure}{0}
\setcounter{page}{1}
\setcounter{section}{0}
\renewcommand{\theequation}{S\arabic{equation}}
\renewcommand{\thefigure}{S\arabic{figure}}
\renewcommand{\thesection}{S\arabic{section}}
\renewcommand{\thepage}{S\arabic{page}}

\begin{center}
    {\large \textbf{Supplementary Material: Digital Quantum Algorithms for Generating and Utilizing Spin Squeezed States}}
\end{center}

\tableofcontents

\section{Preparing spin squeezed states: Proof of Theorem \ref{thm:prepare_sss}}


\subsection{Set up and parameters}

In the main text, we have shown the final (unnormalized) state on data qubits for measurement outcome (\emph{4}) is \begin{equation}
    \ket{\psi_m}= f_m(Z) \ket{+}^{\otimes N},
\end{equation}
where $m=-2^{\ell-1},\cdots,2^{\ell-1}-1$, and \begin{equation}\label{eq:f'm=}
    f_m(Z)=2^{-\frac{\ell}{2}}\sum_{x} \alpha_x \ee^{2\pi\ii x\lr{\phi Z-2^{-\ell}m}} = 2^{-\frac{\ell}{2}} \frac{\mathcal{N}}{\sigma} \sum_{q\in \mathbb{Z}} \ee^{-\frac{\pi}{\sigma^2}\lr{\phi Z-2^{-\ell}m}^2}
\end{equation}
(summed over $x=-2^{\ell-1},\cdots,2^{\ell-1}-1$, same below)
with 
\begin{equation}\label{eq:alphax=}
    \alpha_x = \mathcal{N}\ee^{-\pi\sigma^2x^2}
\end{equation}
normalized by $\sum_x |\alpha_x|^2=1$ so that $\norm{f_m(Z)}\le 1$. The parameters are chosen by \begin{equation}\label{eq:phisigma=}
    \phi=N^{-1/2}/\ell,\quad \sigma = 2^{-\ell}\ell,
\end{equation}
so that \begin{equation}
    \mathcal{N}^2 = \Theta(\sigma)=\Theta(2^{-\ell}\ell)
\end{equation}
from estimation a Gaussian integral. We will use such estimates for Gaussians frequently.

The protocol finally measures the system into $\ket{\psi_m}$ with probability $\braket{\psi_m|\psi_m}$.
We aim to show that with probability \begin{equation}\label{eq:Psuccess}
    P_{\rm success}\ge 1-\frac{1}{\ell},
\end{equation}
the outcome $m$ satisfies \begin{subequations}
    \begin{align}
    \expval{\psi_m}{X}&= \Theta(N), \label{eq:X=N_} \\
    \expval{\psi_m}{\Delta Z^2} &= \order(\ell^{10}) 2^{-2\ell}N. \label{eq:DZ2=xi_}
\end{align}
\end{subequations} 
This proves Theorem \ref{thm:prepare_sss} when combined with Lemma \ref{lem:prepGaussian} (which we prove in the latter Section \ref{sec:Gauss}) on the $\order(\ell)$ protocol depth, by choosing $\ell$ as the minimal integer obeying \begin{equation}\label{eq:ell=xi}
    2^{\ell} \ge \frac{1}{\xi} \log^6 \frac{1}{\xi}.
\end{equation}
Although we will restrict $\ell$ by \begin{equation}\label{eq:c1'}
    2^\ell/\ell \le c_1'\sqrt{N},
\end{equation}
for some sufficiently small constant $c_1'$, \eqref{eq:ell=xi} can be satisfied for any \begin{equation}
    \xi \ge c_1 \frac{\log^6 N}{\sqrt{N}} ,
\end{equation}
with constant $c_1$ determined by $c_1'$.

\subsection{Truncation in $z$}

Expanding $\ket{+}^{\otimes N}$ into Dicke basis $\ket{D_z}:=\ket{D^N_{\frac{N-z}{2}}},z=0,\pm2,\cdots$ (assuming even $N$ for simplicity), \begin{equation}\label{eq:CSS=}
    \ket{+}^{\otimes N} = \sum_z \sqrt{B_z} \ket{D_z}, \where B_z = 2^{-N} {N \choose \frac{N+z}{2}}=\sqrt{\frac{2}{\pi N}}\ee^{-\frac{z^2}{2N}}\lr{1+\order(N^{-1})},
\end{equation}
for all $|z|=\order(N^{2/3})$. $B_z$ satisfies Lipschitz condition \begin{equation}\label{eq:Lipshitz}
    |B_z - B_{z'}|\le |z-z'|/N,
\end{equation}
from the Lipschitz condition of the Gaussian $\ee^{-\frac{z^2}{2N}}$, where the $\order(N^{-1})$ part in \eqref{eq:CSS=} is subdominant comparing to the right hand side of \eqref{eq:Lipshitz}.

Denoting \begin{equation}\label{eq:Zq=}
    \mathcal{Z}_q=\{z\in \mathbb{Z}_{\rm even}: q-\frac{1}{2}\le \phi z<q+\frac{1}{2}\},
\end{equation}
The classical function $f_m(z)$ from \eqref{eq:f'm=} is concentrated at $z=z_m$ where \begin{equation}\label{eq:zm=}
     z_m=2^{-\ell}\phi^{-1}m
\end{equation}  
in the $q=0$ period $\mathcal{Z}_0$. 
We can focus on this single period, because we will only consider two observables \begin{equation}\label{eq:gz=}
    g(Z) = I \text{ or } (Z-z_m)^2/N,
\end{equation}
whose expectation value is \begin{align}\label{eq:g-g0}
    \bra{\psi_m} g(Z) \ket{\psi_m}-\bra{\psi_m} g(Z) \ket{\psi_m}_0 &= \sum_{q\in \mathbb{Z}:q\neq 0} \sum_{z\in \mathcal{Z}_q} \abs{f_m(z)}^2 B_z g(z) \nonumber\\
    &= 2\sum_{q=1,2,\cdots} \order(\phi^{-1}N^{-1/2}) \poly{q\phi^{-1}N^{-1/2}} \exp\mlr{-\frac{(q-\frac{1}{2})^2}{2N\phi^2} } + \exp\mlr{-\Omega(N^{1/3})} \nonumber\\
    &= 2\sum_{q=1,2,\cdots} \order(\ell) \poly{q\ell} \exp\mlr{-\frac{(q-\frac{1}{2})^2\ell^2}{2} } + \exp\mlr{-\Omega(N^{1/3})}  \nonumber\\
    &= \exp\mlr{-\Omega(\ell^2)},
\end{align}
using \eqref{eq:CSS=} and \eqref{eq:phisigma=}, where \begin{equation}\label{eq:g0=}
    \bra{\psi_m} g(Z) \ket{\psi_m}_0:= \sum_{z\in \mathcal{Z}_0} \abs{f_m(z)}^2 B_z g(z).
\end{equation}
The $\exp\mlr{-\Omega(N^{1/3})}$ term in \eqref{eq:g-g0} comes from the $|z|=\Omega(N^{2/3})$ contribution, where the exponentially small $B_z$ dominates the other factors scaling at most polynomially.

\subsection{$Z$ variance}
To calculate $\bra{\psi_m} g(Z) \ket{\psi_m}_0$, observe that for the second choice in \eqref{eq:gz=}, \begin{align}
    \abs{B_z g(z)-B_{z_m}g(z_m)} \le \lr{\max_z |B_z|} (z-z_m)^2/N \le \frac{1}{\sqrt{N}} \lr{\frac{z-z_m}{\sqrt{N}}}^c 
\end{align}
with $c=2$. The first choice in \eqref{eq:gz=} also satisfies the above equation with $c=1$.
As a result, letting \begin{equation}
    F_m=\sum_{z\in \mathcal{Z}_0} \abs{f_m(z)}^2,
\end{equation}
we have \begin{align}\label{eq:gz-gzm}
    \bra{\psi_m} g(Z) \ket{\psi_m}_0 - B_{z_m}g(z_m)F_m &= \sum_{z\in \mathcal{Z}_0} \abs{f_m(z)}^2 \mlr{B_z g(z) - B_{z_m}g(z_m)} \nonumber\\
    &\le \sum_{z\in \mathcal{Z}_0} \abs{f_m(z)}^2 \frac{1}{\sqrt{N}} \lr{\frac{z-z_m}{\sqrt{N}}}^c \nonumber\\
    &= N^{-(1+c)/2} \order\lr{(\sigma/\phi)^{1+c}} \nonumber\\
    &= \order(\ell^{2(1+c)}) 2^{-(1+c)\ell}.
\end{align}
We have used $f_m(z)$ in $\mathcal{Z}_0$ concentrates in window $|z-z_m|=\order(\sigma/\phi)$ for all \begin{equation}\label{eq:m<}
    |m|\le 2^{\ell-2}
\end{equation}
so that $z_m$ is not too close to the boundary of $\mathcal{Z}_0$. The probability $P_{|m|\le 2^{\ell-2}}$ to measure this range \eqref{eq:m<} is close to $1$: \begin{align}
    &1-P_{|m|\le 2^{\ell-2}}=1-\sum_{m:|m|\le 2^{\ell-2}} \braket{\psi_m|\psi_m} \nonumber\\
    &\quad =  \sum_{m:|m|> 2^{\ell-2}} \braket{\psi_m|\psi_m} = \sum_{m:|m|> 2^{\ell-2}} \sum_{z\in \mathcal{Z}_0} \abs{f_m(z)}^2 B_z \nonumber\\
    &\quad = \sum_{z\in \mathcal{Z}_0:|\phi z|\le 1/4} B_z \sum_{m:|m|> 2^{\ell-2}} \ee^{-\Omega(\sigma^{-2})} + \sum_{z\in \mathcal{Z}_0:|\phi z|> 1/4} \frac{1}{\sqrt{N}} \ee^{-\Omega(\phi^{-2}/N)} \sum_{m:|m|> 2^{\ell-2}} \order(1) \hspace{30em}\mytag{\text{Using \eqref{eq:f'm=} and \eqref{eq:CSS=}}} \nonumber\\
    &\quad =\order(1)\cdot 2^{\ell}\cdot \ee^{-\Omega(\sigma^{-2})} + \order(\phi^{-1}/\sqrt{N})\ee^{-\Omega(\ell^2)}2^{\ell} \nonumber\\
    &\quad = \ee^{-\Omega(\ell^2)}. \label{eq:Pm<}
\end{align}

\comment{
The concentration of  also implies \begin{equation}\label{eq:Fm=}
    F_m = \Theta(\phi^{-1}\sigma)=\Theta(\sqrt{N}\ell^{1+2\kappa}2^{-\ell}),
\end{equation}
using \eqref{eq:Nf=1}.
}

Combining \eqref{eq:gz-gzm} with \eqref{eq:g-g0}, \begin{equation}
    \bra{\psi_m} g(Z) \ket{\psi_m} = B_{z_m}g(z_m)F_m + \order(\ell^{3(1+c)}) 2^{-(1+c)\ell}.
\end{equation}
For the two choices \eqref{eq:gz=} with $c=1,2$ respectively, we get \begin{subequations}
    \begin{align}
    \braket{\psi_m|\psi_m} &= B_{z_m} F_m + \order(\ell^{6}) 2^{-2\ell}, \\
    \bra{\psi_m}(Z-z_m)^2\ket{\psi_m}&= N\order(\ell^{9})  2^{-3\ell}. \label{eq:Z2psim=}
\end{align}
\end{subequations} 

The $Z$ variance is then bounded \begin{equation}\label{eq:Z2m<}
    \alr{\Delta Z^2}_{\psi_m} \le \frac{\bra{\psi_m}(Z-z_m)^2\ket{\psi_m}}{\braket{\psi_m|\psi_m}} = N \order(\ell^{10})  2^{-2\ell},
\end{equation}
as long as \begin{equation}\label{eq:Pm=2-l}
    \braket{\psi_m|\psi_m} \ge 2^{-\ell}/\ell.
\end{equation}
Since the probability of obtaining $m$ that violates \eqref{eq:Pm=2-l} is \begin{equation}
    \sum_{m:\braket{\psi_m|\psi_m}\ll 2^{-\ell}/\ell} \braket{\psi_m|\psi_m} \le \sum_{m:\braket{\psi_m|\psi_m}\ll 2^{-\ell}/\ell} 2^{-\ell}/\ell \le 1/\ell,
\end{equation}
where the $m$s includes all of those outside range \eqref{eq:m<} due to \eqref{eq:Pm<}, with success probability \eqref{eq:Psuccess} $m$ indeed satisfies \eqref{eq:Pm=2-l} and \eqref{eq:Z2m<}. This establishes \eqref{eq:DZ2=xi_}.

\subsection{$X$ polarization}
Before the inverse QFT stage of the protocol, the state on $A$ $+$ data qubits is \begin{equation}\label{eq:Psi=thetax}
    \ket{\Psi}=\sum_x \alpha_x \ket{x}\otimes \ket{\theta_x}^{\otimes N},
\end{equation}
where \begin{equation}\label{eq:thetax}
    \ket{\theta_x} =\cos \theta_x \ket{+} + \sin \theta_x \ket{-}
\end{equation}
with \begin{equation}\label{eq:thetax<}
    |\theta_x|=|4\pi\phi x|\le \pi 2^{\ell+1}/(\ell \sqrt{N}) \le 2\pi c_1'
\end{equation}
from \eqref{eq:c1'}.
Note that a phase may be ignored in \eqref{eq:thetax} that is irrelevant for the following discussion. For a given $x$, observable $X$ thus follows a binomial distribution \begin{equation}
    \mathrm{Prob}[X=N-2k] = {N\choose k} \cos^{2k} \theta_x \sin^{2(N-k)} \theta_x,
\end{equation}
that concentrates exponentially at $X=N\cos\theta_x$. With a sufficiently small constant $c'_1$, $|\theta_x|\ll 1$ so that \begin{align}
    \mathrm{Prob}[X<N/2]=\sum_{k>N/4}\mathrm{Prob}[X=N-2k] = \ee^{-\Omega(N)}.
\end{align}
This holds for the whole state $\ket{\Psi}$ because observable $x$ is independent with $X$. 

As a result, for every $m$ with probability \eqref{eq:Pm=2-l} that is at most polynomially small in $N$, the normalized $\ket{\psi_m}$ is supported in the $X\ge N/2$ subspace up to $\ee^{-\Omega(N)}$ error, which establishes \eqref{eq:X=N_}. This finishes the proof of Theorem \ref{thm:prepare_sss}.

\subsection{$Y$ variance and noise robustness for metrology}

From \eqref{eq:Psi=thetax}, we also get for each $x$, \begin{equation}\label{eq:ProbY>}
    \mathrm{Prob}[|Y|>y] = \ee^{-\Omega(y^2/N)}
\end{equation}
for any \begin{equation}\label{eq:y>Ntheta}
    y\ge N\theta_x\cdot \ell\gg N\theta_x
\end{equation}
from bounding the binomial distribution. Similar to the arguments above, for every $m$ with probability \eqref{eq:Pm=2-l} that is much larger than \eqref{eq:ProbY>}, the normalized $\ket{\psi_m}$ is supported in the $Y$ subspace $|Y|\le y$ up to $\ee^{-\Omega(y^2/N)}$ error, for any $y\ge \pi\sqrt{N}2^{\ell+1}$ from \eqref{eq:thetax<}. By choosing a growing set of $y_1<y_2<\cdots$ in this range and summing over the fast decaying probabilities from each $Y$ windows $[y_1,y_2],[y_2,y_3],\cdots$, we have \begin{equation}\label{eq:Y2=NO}
    \alr{\Delta Y^2}_{\psi_m} = N\cdot \order(2^\ell) = N\cdot \widetilde{\order}(\xi^{-2})
\end{equation}
using \eqref{eq:ell=xi}, where $\widetilde{\order}$ hides a $\mathrm{polylog}(\xi)$ factor. This almost saturates the Heisenberg uncertainty relation $\alr{\Delta Y^2}=\Omega(N/\xi^2)$ for SSSs, which make them optimally robust for noisy metrology \cite{sss_noise24}. The idea is that noise may mix the polarization along different directions with the signal direction $Z$, so it is desirable to have fluctuation along $Y$ as small as possible. If \eqref{eq:Y2=NO} holds, then the squeezing parameter remains $\widetilde{\order}(\xi)$ if the SSS undergoes arbitrary single-qubit noise with strength $p\ll\xi^2$ per qubit.

\section{Preparing Dicke states deterministically: Proof of Theorem \ref{thm:dicke}}

\subsection{Theorem \ref{thm:prepare_sss} as a transformation of Gaussian states}

In this section, we use Stinespring dilation to represent the state of different measurement outcomes into one state. For example, the dilated protocol $\mathcal{U}$ of Theorem \ref{thm:prepare_sss} followed by resetting $A$ to the all-zero state $\ket{0}_A$, prepares the state \begin{equation}\label{eq:USS=}
    \mathcal{U} \ket{0}_A\ket{+}^{\otimes N} = \ket{0}_A\otimes \sum_{m\in \mathcal{M}} \ket{m}_{\rm C}\otimes \ket{\psi_m} + \order\lr{\frac{1}{\log(1/\xi)}},
\end{equation}
where $\ket{m}_{\rm C}$ represents state $m$ in a classical register $\rm C$, and $\mathcal{U}$ is an isometry $\mathcal{U}^\dagger \mathcal{U}=I$ that maps to an enlarged Hilbert space including the classical register as a subsystem. $\mathcal{M}$ is the set of successful measurement outcomes, so we also includes an error term $\order\lr{\frac{1}{\log(1/\xi)}}$ corresponding to the failure probability, which strictly speaking means $\norm{\ket{\Psi}-\ket{\Psi'}}=\order\lr{\frac{1}{\log(1/\xi)}}$ where $\ket{\Psi},\ket{\Psi'}$ are the left (first term of right) hand side of \eqref{eq:USS=}.

We would like to iterate protocol $\mathcal{U}$ multiple times; however, the initial and final states in \eqref{eq:USS=} are not in the same form. We aim to rewrite \eqref{eq:USS=} as a transformation of Gaussian states. First, the initial state can be approximated by \begin{equation}\label{eq:CSS=Gauss}
    \ket{+}^{\otimes N} = \ket{0,\sqrt{N}} + \order\lr{\frac{1}{N}},
\end{equation}
using \eqref{eq:CSS=}, where \begin{equation}\label{eq:mugamma=}
    \ket{\mu;\gamma}:= \mathcal{N}_{\mu;\gamma} \sum_z \ee^{-\frac{(z-\mu)^2}{4\gamma^2}} \ket{D_z},
\end{equation}
with $\mathcal{N}_{\mu;\gamma}$ being normalization such that $\norm{\ket{\mu;\gamma}}=1$. Recall that $z$ sums over $z=0,\pm2,\cdots,\pm N$. We will only consider cases with $\gamma\gg 1$
so that $\ket{\mu;\gamma}$ is supported in many Dicke states.

We consider protocol $\mathcal{U}$ acting on a general Gaussian initial state: \begin{align}
    \mathcal{U}_\phi \ket{0}_A \ket{0;\gamma_0} = \ket{0}_A\otimes \sum_{m} \ket{m}_{\rm C}\otimes F_m(\phi Z)\ket{0;\gamma_0} ,
\end{align} 
where $\mathcal{U}_\phi$ is the protocol of Theorem \ref{thm:prepare_sss} with a different parameter \begin{equation}\label{eq:phi=gamma0}
    \phi = 1/(\ell \gamma_0),
\end{equation}
(so $\gamma_0=\sqrt{N}$ returns to the previous case \eqref{eq:phisigma=}; the other parameters like $\sigma$ are the same as before) and $F_m(\phi Z)=f_m(Z)$ \eqref{eq:fm=} with the $\phi$ dependence now explicit. 

Let $F'_m(\phi Z)=\mathcal{N}_f\exp\mlr{-\frac{\pi}{\sigma^2}\lr{\phi Z-2^{-\ell}m}^2}$ be the $q=0$ term of \eqref{eq:f'm=} with $\mathcal{N}_f=2^{-\frac{\ell}{2}} \frac{\mathcal{N}}{\sigma}$, \begin{align}\label{eq:F'-F'}
    \norm{\mlr{F_m(\phi Z) - F'_m(\phi Z)}\ket{0;\gamma_0}} &\le \mathcal{N}_f \sum_{q\in \mathbb{Z}:q\neq 0}\norm{\exp\mlr{-\frac{\pi}{\sigma^2}\lr{\phi Z-2^{-\ell}m+q}^2} \ket{0;\gamma_0} } \nonumber\\
    &\le \mathcal{N}_f \sum_{q\in \mathbb{Z}:q\neq 0} \sum_{q'\in \mathbb{Z}} \norm{\exp\mlr{-\frac{\pi}{\sigma^2}\lr{\phi Z-2^{-\ell}m+q}^2}\mathcal{P}_{q'} \ket{0;\gamma_0} } \nonumber\\
    &= \mathcal{N}_f \sum_{q\in \mathbb{Z}:q\neq 0} \sum_{q'\in \mathbb{Z}} \ee^{-\Omega\lr{\frac{(q-q')^2}{\sigma^2}}} \ee^{-\Omega\lr{\frac{(q')^2}{\phi^2\gamma_0^2}}} \nonumber\\
    &= \ee^{-\Omega(\ell^2)}.
\end{align}
The third line here uses $\norm{\mathcal{P}_{q'}\ket{0;\gamma_0}}=\exp\mlr{-\Omega\lr{\frac{(q')^2}{\phi^2\gamma_0^2}}}$ from \eqref{eq:mugamma=} where $\mathcal{P}_{q'}$ projects onto $z\in \mathcal{Z}_{q'}$ subspace \eqref{eq:Zq=}, and assumed $m$ in range \eqref{eq:m<} to bound the filter function. The last line uses \eqref{eq:phi=gamma0} and \eqref{eq:phisigma=} so that the sum is dominated by $q=q'=\pm1$. Similar to \eqref{eq:Pm<}, we bound the probability of obtaining range \eqref{eq:m<}: \begin{align}
    \sum_{m:|m|>2^{\ell-2}}\norm{F'_m(\phi Z)\ket{0;\gamma_0}}^2 &\le \sum_{m:|m|>2^{\ell-2}}\mlr{\norm{F'_m(\phi Z)\mathcal{P}_0\ket{0;\gamma_0}}^2 +\sum_{q\in \mathbb{Z}:q\neq 0} \order(1) \norm{\mathcal{P}_q \ket{0;\gamma_0}}^2 }\hspace{20em}\nonumber\\
    &= \sum_{m:|m|>2^{\ell-2}} \ee^{-\Omega(\sigma^{-2})}\norm{\mathcal{P}_{|\phi z|\le 2^{-3}}\ket{0;\gamma_0} }^2 + \order(1)\cdot \norm{\mathcal{P}_{|\phi z|> 2^{-3}}\ket{0;\gamma_0} }^2 + \ee^{-\Omega(\ell^2)} \nonumber\\
    &= \ee^{-\Omega(\ell^2)},
\end{align}
where we have used $\mathcal{P}_{|\phi z|\le 2^{\ell-3}}$ to project onto the $z$ subspace in the subscript, and $\norm{\mathcal{P}_{|\phi z|\le 2^{-3}}\ket{0;\gamma_0} }\le 1, \norm{\mathcal{P}_{|\phi z|> 2^{-3}}\ket{0;\gamma_0} }=\ee^{-\Omega(\ell^2)}$. Therefore we can restrict the sum over $m$ in \eqref{eq:U=F'} to \eqref{eq:m<}, for which we can further apply \eqref{eq:F'-F'}: \begin{equation}\label{eq:U=F'}
    \mathcal{U}_\phi \ket{0}_A \ket{0;\gamma_0} = \ket{0}_A\otimes \sum_{m:|m|\le 2^{\ell-2}} \ket{m}_{\rm C}\otimes F'_m(\phi Z)\ket{0;\gamma_0} + \ee^{-\Omega(\ell^2)}.
\end{equation}

Finally, \begin{align}
    F'_m(\phi Z)\ket{0;\gamma_0} = \mathcal{N}_{\mu;\gamma} \mathcal{N}_f \sum_z \ee^{-\frac{\pi}{\sigma^2}\lr{\phi z-2^{-\ell}m}^2} \ee^{-\frac{z^2}{4\gamma_0^2}} \ket{D_z} =\sqrt{P_m} \ket{\mu_m; \gamma},
\end{align}
where $P_m\ge 0$ is normalization, and \begin{equation}\label{eq:mum<}
    |\mu_m|\le \phi^{-1}=\ell \gamma_0, \quad \gamma=\lr{\frac{4\pi \phi^2}{\sigma^2} + \frac{1}{\gamma_0^2}}^{-1/2} =\Theta(\ell \sigma \gamma_0).
\end{equation}
\eqref{eq:U=F'} then becomes \begin{equation}\label{eq:U=Gauss}
    \mathcal{U}_\phi \ket{0}_A \ket{0;\gamma_0} = \ket{0}_A\otimes \sum_{m:|m|\le 2^{\ell-2}} \sqrt{P_m} \ket{m}_{\rm C}\otimes \ket{\mu_m;\gamma} + \ee^{-\Omega(\ell^2)}
\end{equation}
which transforms initial Gaussian state $\ket{0;\gamma_0}$ to Gaussian states with much smaller standard deviation $\gamma\ll \gamma_0$, up to small errors.

\subsection{Rotating Gaussian states close to the equator}

For each $\ket{\mu_m;\gamma}$ in \eqref{eq:U=Gauss} satisfying \eqref{eq:mum<} and \begin{equation}\label{eq:gamma<sqrtN}
    \gamma_0=\order(\sqrt{N}),
\end{equation}
there exists a rotation angle $\theta_m$ such that \begin{equation}\label{eq:thetam=}
    \ee^{\ii \theta_m Y}\ket{\mu_m;\gamma} = \ket{0;\gamma} + \order(N^{-1/3})+ \order(\gamma_0)\ee^{-\Omega\lr{2^\ell}} +\order(2^{3\ell}/\gamma_0),
\end{equation}
where $Y=\sum_i Y_i$.

\begin{proof}
The spin operator\begin{equation}
    Y = -\ii \sum_z \sqrt{\frac{N}{2}\lr{\frac{N}{2}+1}-\frac{z}{2}\lr{\frac{z}{2}+1}} \ket{D_{z-2}} \bra{D_z} + \mathrm{H.c.},
\end{equation}
can be approximated by \begin{equation}\label{eq:PYP=}
    \norm{\lr{Y-\frac{N}{2}T}\mathcal{P}_{|z|\le w}} = \order\lr{\frac{w^2}{N}}
\end{equation}
restricting to $w$-width window near the equator $z=0$. The operator \begin{equation}
    T=-\ii \sum_z^\infty \ket{D_{z-2}} \bra{D_z}+\mathrm{H.c.}
\end{equation}
acts on an enlarged Hilbert space containing all $z\in \mathbb{Z}_{\rm even}$, where $\sum_z^\infty$ sums over this set. It effectively generates translation for low-momentum states $\ket{\psi}=\sum_z^\infty \psi(z)\ket{D_z}=\int^\pi_{-\pi}\frac{\dd k}{2\pi} \hat{\psi}(k)\ket{k}_{\rm mo}$:  \begin{align}\label{eq:translate}
    \ee^{\ii \eta T} \ket{\psi} &= \int^\pi_{-\pi}\frac{\dd k}{2\pi} \hat{\psi}(k)\ee^{\ii \eta T}\ket{k}_{\rm mo} \nonumber\\
    &= \int^\pi_{-\pi}\frac{\dd k}{2\pi} \hat{\psi}(k)\ee^{-\ii 2\eta[k+\order(k^3)] }\ket{k}_{\rm mo} \nonumber\\
    &= \int^\pi_{-\pi}\frac{\dd k}{2\pi} \hat{\psi}(k)\ee^{-\ii 2\eta k }\ket{k}_{\rm mo}+ \int^\pi_{-\pi}\frac{\dd k}{2\pi} \hat{\psi}(k)\ee^{-\ii 2\eta k }\order(\eta k^3)\ket{k}_{\rm mo} \nonumber\\ 
    &= \sum_z^\infty \psi(z+4\eta)\ket{D_z} + \order(\eta)\lr{\int^\pi_{-\pi}\frac{\dd k}{2\pi} \abs{\hat{\psi}(k)}^2 k^6}^{1/2},
\end{align}
where we have used $\ket{k}_{\rm mo}=\sum_z^\infty \ee^{\ii k(z/2)}\ket{D_z}$ is eigenstate of $T$ with eigenvalue $\ii(\ee^{\ii k}-\ee^{-\ii k})=-2k+\order(k^3)$. For Gaussian $\ket{\psi}=\ket{\mu;\gamma}_\infty$ that is \eqref{eq:mugamma=} with $z$ summed over all even integers, \begin{align}
    \hat{\psi}(k) &=\mathcal{N}_{\mu;\gamma} \sum_z^\infty \ee^{-\ii k(z/2)}\ee^{-\frac{(z-\mu)^2}{4\gamma^2}} \nonumber\\
    &=\mathcal{N}_{\mu;\gamma}\sqrt{\pi}\gamma \sum_{n\in \mathbb{Z}}\ee^{-\ii(k+2\pi n)\mu} \ee^{-\frac{\gamma^2}{4}(k+2\pi n)^2}
\end{align}
is concentrated in $|k|=\order(\gamma^{-1})$ using Poisson summation formula, so \eqref{eq:translate} becomes \begin{equation}\label{eq:translate1}
    \ee^{\ii \eta T} \ket{\mu;\gamma}_\infty = \ket{\mu+4\eta;\gamma}_\infty + \order(\eta/\gamma^{3}).
\end{equation}

For \begin{equation}
    |\mu|\le w/2\ll N, 
\end{equation}
the projection error in the $w$-width window is small
\begin{equation}\label{eq:Pmugamma=}
    \mathcal{P}_{|z|\le w}\ket{\mu;\gamma} = \ket{\mu;\gamma} + \ee^{-\Omega\lr{\frac{w^2}{\gamma^2}}} = \ket{\mu;\gamma}_\infty + \ee^{-\Omega\lr{\frac{w^2}{\gamma^2}}},
\end{equation}
so that
\begin{align}
    &\lr{\ee^{\ii \theta Y} - \ee^{\ii \frac{\theta N}{2}T} } \ket{\mu_m;\gamma}_\infty =\ii \int^{\theta N/2}_0 \dd \eta \ee^{\ii \lr{\theta -\frac{2\eta}{N}}Y} \lr{\frac{2}{N}Y-T} \ee^{\ii \eta T} \ket{\mu_m;\gamma}_\infty \hspace{20em}\mytag{\text{Duhamel identity}} \\
    &\qquad \le\int^{\theta N/2}_0 \dd \eta \norm{\lr{\frac{2}{N}Y-T} \ket{\mu;\gamma}_\infty} + \int^{\theta N/2}_0 \dd \eta \order(\eta/\gamma^3) \mytag{\text{Using \eqref{eq:translate1}, }\norm{T}\le 2\text{ and }\norm{Y}= N}\\
    &\qquad \le\int^{\theta N/2}_0 \dd \eta \norm{\lr{\frac{2}{N}Y-T} \mathcal{P}_{|z|\le w}} + \int^{\theta N/2}_0 \dd \eta \ee^{-\Omega\lr{\frac{w^2}{\gamma^2}}} +\order(\mu_m^2/\gamma^3)  \mytag{\text{Using \eqref{eq:Pmugamma=} and }\theta=-\mu_m/(2N)} \\ 
    &\qquad = \order(|\mu_m|w^2/N^2)+ \order(|\mu_m|)\ee^{-\Omega\lr{\frac{w^2}{\gamma^2}}} +\order(\mu_m^2/\gamma^3)\mytag{\text{Using \eqref{eq:PYP=}}}\\ \label{eq:Y-T<}
\end{align}
Here we have used $|\mu|\le |\mu_m|$ from the choice of $\theta$, and the last line follows from \eqref{eq:Pmugamma=} assuming \begin{equation}\label{eq:mum<w}
    |\mu_m|\le w/2.
\end{equation}

Combining \eqref{eq:Y-T<} with \eqref{eq:Pmugamma=} and \eqref{eq:translate1} yields 
\begin{align}\label{eq:mu_to_0}
    \ee^{\ii \theta Y} \ket{\mu_m;\gamma} &= \ee^{\ii \theta Y} \ket{\mu_m;\gamma}_\infty +\ee^{-\Omega\lr{\frac{w^2}{\gamma^2}}} \nonumber\\
    &= \ket{0;\gamma} + \order(\ell \gamma_0 w^2/N^2)+ \order(\ell \gamma_0)\ee^{-\Omega\lr{\frac{w^2}{\gamma^2}}} +\order(2^{3\ell}/\gamma_0)
\end{align}
where we have used \eqref{eq:phi=gamma0},\eqref{eq:mum<} and \eqref{eq:phisigma=}. \eqref{eq:mu_to_0} becomes \eqref{eq:thetam=} by setting \begin{equation}
    w=2/\phi = 2\ell \gamma_0
\end{equation}
that satisfies \eqref{eq:mum<}, and using \eqref{eq:gamma<sqrtN} and $\ell=\order(\log N)$.
\end{proof}

\subsection{The whole protocol}
We use measurement outcome $m$ in to \eqref{eq:U=Gauss} adaptively apply rotation $\theta_m$ in \eqref{eq:thetam=}. This procedure $\mathcal{R}_\phi$ ensures that \begin{equation}
    \mathcal{R}\mathcal{U}_\phi \ket{0}_A \ket{0;\gamma_0} = \ket{0}_A\ket{0;\gamma}\otimes \sum_{m:|m|\le 2^{\ell-2}} \sqrt{P_m} \ket{m}_{\rm C}  + \epsilon
\end{equation}
where \begin{equation}
    \epsilon = \ee^{-\Omega(\ell^2)}+\order(N^{-1/3})+ \order(\gamma_0)\ee^{-\Omega\lr{2^\ell}} +\order(2^{3\ell}/\gamma_0).
\end{equation}

Starting from initial state \eqref{eq:CSS=Gauss}, we apply $\mathcal{R}_{\phi_1}\mathcal{U}_{\phi_1}$ with \begin{equation}\label{eq:gamma0=sqrtN}
    \gamma_0=\sqrt{N}
\end{equation} 
to squeeze to $\ket{0;\gamma_1}$, and then apply another $\mathcal{R}_{\phi_2}\mathcal{U}_{\phi_2}$ to squeeze to $\ket{0;\gamma_2}$, and so on: \begin{equation}\label{eq:RURU}
    \ket{\Psi}=\mathcal{R}_n\mathcal{U}_{\phi_n}\cdots \mathcal{R}_2\mathcal{U}_{\phi_2}\mathcal{R}_1\mathcal{U}_{\phi_1} \ket{0}_A\ket{+}^{\otimes N} = \ket{0}_A \ket{0;\gamma_n}\otimes \ket{\text{outcome}}_{\rm C} + \order(N^{-1})+\sum_{k=1}^n \epsilon_k.
\end{equation} 
Here $\ket{\text{outcome}}_{\rm C}$ is some unnormalized state for the classical register for the past measurement outcomes.
The only free parameters is $\ell_1,\cdots,\ell_n$, because $\gamma_k$ is determined by \eqref{eq:mum<}: \begin{equation}\label{eq:gammak=gammak-1}
    \gamma_k = \Theta(\ell_k^3 2^{-\ell_k}\gamma_{k-1} ).
\end{equation}
The error is \begin{equation}\label{eq:epsk=}
    \epsilon_k = \ee^{-\Omega(\ell_k^2)}+\order(N^{-1/3})+ \order(\gamma_{k-1})\ee^{-\Omega\lr{2^{\ell_k} }} +\order(2^{3\ell_k}/\gamma_{k-1}).
\end{equation}

We choose \begin{equation}\label{eq:lk=}
    \ell_k = \floor{\frac{1}{8} \log(\gamma_{k-1}) } +1, 
\end{equation}
so that \eqref{eq:gammak=gammak-1} becomes \begin{equation}
    \log(\gamma_k) = \frac{7}{8}\log(\gamma_{k-1}) + \order\lr{\log\log \gamma_{k-1}}\in \mlr{\frac{4}{5}\log(\gamma_{k-1}),\frac{9}{10}\log(\gamma_{k-1})}, 
\end{equation}
for $\gamma_{k-1}\gg 1$.
Iterating $k$ from \eqref{eq:gamma0=sqrtN}, there exists \begin{equation}
    n=\Theta(\log \log N).
\end{equation} 
such that $\gamma_n$ is a sufficiently large constant independent of $N$, and \begin{equation}\label{eq:gammak>gamman}
    \log\gamma_k \ge \lr{\frac{10}{9}}^{n-k} \log\gamma_n.
\end{equation}

Using \eqref{eq:epsk=} and \eqref{eq:lk=}, the error with $k\le n$ is \begin{align}
    \epsilon_k &= \ee^{-\Omega(\log \gamma_{k-1})^2} + \order(\gamma_{k-1})\ee^{-\Omega(\gamma_{k-1}^{1/8})} + \order(\gamma_{k-1}^{-5/8}) \nonumber\\
    &= \order(\ee^{-\frac{5}{8}\log(\gamma_{k-1})}).
\end{align}
The total error in \eqref{eq:RURU} is then \begin{equation}\label{eq:total_error}
    \widetilde{\epsilon}:=\order(N^{-1})+\sum_{k=1}^n \epsilon_k = \sum_{k=1}^n \order(\ee^{-\frac{5}{8}\lr{\frac{5}{4}}^{n-k} \log(\gamma_n)}) = \order(\gamma_n^{-5/8}).
\end{equation}

We finally measure $Z$ in state \eqref{eq:RURU}, which can be achieved by $\order(\log N)$-depth circuit \cite{Dicke_circuit24}. The probability $P_{\rm success}=\norm{\mathcal{P}_{z=0}\ket{\Psi}}^2$ to measure the target $z=0$ is \begin{equation}
    P_{\rm success}=\mlr{\Theta(\gamma_n^{-1/2}) - \order(\widetilde{\epsilon})}^2 = \Theta(\gamma_n^{-1})
\end{equation}
using \eqref{eq:total_error} and \eqref{eq:mugamma=}, which is a $>0$ constant. The whole protocol has depth \begin{equation}
    \order(\log N)+\sum_{k=1}^n \order(\ell_k) = \order(\log N) + \order(\ell_1) = \order(\log N)
\end{equation}
using \eqref{eq:lk=} and \eqref{eq:gammak>gamman}.


\section{Gaussian states: Proof of Lemma \ref{lem:prepGaussian} and a matrix product state representation}
\label{sec:Gauss}

\subsection{Upper bound on preparation circuit depth}
Let $\mathcal{X}_\ell=\{-2^{\ell-1},\cdots,2^{\ell-1}-1\}$ be a equispaced grid and $|\mathcal{X}_\ell|=2^\ell=L$. For $x\in\mathcal{X}_\ell$, define $f_x=e^{-\pi\sigma^2 x^2}$.  And we define a parameter
\begin{equation}
B=2^{2(\ell-1)}\pi\sigma^2\geq 1,    
\end{equation}
which relates the Gaussian width and the grid width. Therefore, $0\leq\pi\sigma^2x^2\leq B$ and $f_x\geq e^{-B}$, $\forall x\in\mathcal{X}_\ell$. The normalized target state on $\ell$ qubits is
\begin{equation}
\label{eq:gaussianstate}
    |\alpha\rangle=\frac{1}{\|f\|_2}\sum_{x\in\mathcal{X}_\ell}f_x|x\rangle.
\end{equation}

We consider the case $B=\Omega(\ell^2)$, i.e. $\sigma=\Omega(2^{-\ell}\ell)$. If $B=\Omega(2^{2\ell}/\ell^2)$, we can truncate the Gaussian to smaller system effective supported on $\log\ell$ qubits and efficiently prepare it with state error $e^{-\Omega(\ell^2)}$, as we will show later. If $\Theta(\ell^2)\leq B<\Theta(2^{2\ell}/\ell^2)$, the Gaussian width satisfies $\Theta(2^\ell/\ell)\geq 1/\sigma>\Theta(\ell)$. For $B=\omega(\ell^2)$, we can always choose a smaller number of qubits $\ell'<\ell$ so that the same Gaussian, when viewed on that smaller system, has an effective parameter $B'=\Theta(\ell^2)$. (For $B=\Theta(\ell^2)$, we just have $B'=B$.) According to the definition of the effective parameter, we have
\begin{equation}
    \frac{B}{2^{2\ell}}=\frac{B'}{2^{2\ell'}}=\frac{\pi\sigma^2}{4},
\end{equation}
so
\begin{equation}
    B'=B2^{-2(\ell-\ell')}.
\end{equation}
Now choose $B'=\Theta(\ell^2)$ and substitute it into the above equation, we have
\begin{equation}
    \ell'=\ell-\log_2\frac{\sqrt{B}}{\ell}+\order(1).
\end{equation}
From the range of $B$, we get
\begin{equation}
\label{eq:lprimerange}
    \ell\geq\ell'>2\log_2\ell.
\end{equation}
Denote the truncated state on $\ell'$ qubits as $|\alpha_{\ell'}\rangle$, the state error is
\begin{equation}
    \||\alpha\rangle-|\alpha_{\ell'}\rangle\|_2\leq\sqrt{\frac{\int^{\infty}_{2^{\ell'-1}}e^{-2\pi\sigma^2x^2}\mathrm{d}x}{\sum_{x\in\mathcal{X}_\ell}e^{-2\pi\sigma^2x^2}}}\leq \sqrt{\frac{\order(\frac{e^{-2B'}}{\sigma\sqrt{B'}})}{\Omega(\frac{1}{\sigma})}}=\order (B'^{-\frac{1}{4}}e^{-B'})=e^{-\Omega(\ell^2)},
\end{equation}
where we have used facts similar to Eq.~\eqref{eq:truncate} and Eq.~\eqref{eq:norm}.

The main fact for Lemma \ref{lem:prepGaussian} that we will utilize is that a Gaussian with large width $\Theta(2^{\ell'}/\sqrt{B'})$ after a Fourier transform becomes a Gaussian with small width $\Theta(\sqrt{B'})$, which will be easier to prepare.

Consider the periodization of the Gaussian
\begin{equation}
    G(x)=\sum_{\nu\in\mathbb{Z}}e^{-\pi\sigma^2(x+\nu L')^2},
\end{equation}
where $L'=2^{\ell'}$. We apply the Poisson summation formula, \begin{equation}\label{eq:x2=y2}
    G(x)= \sum_{\nu\in \mathbb{Z}}\ee^{-\pi \widetilde{\sigma}^2(\widetilde{x}+\nu)^2} = \frac{1}{\widetilde{\sigma}} \sum_{r\in \mathbb{Z}} \ee^{-\frac{\pi}{\widetilde{\sigma}^2}r^2} \ee^{2\pi\ii\widetilde{x}r}=\frac{1}{L'\sigma}\sum_{r\in\mathbb{Z}}e^{-\frac{\pi r^2}{L'^2\sigma^2}}e^{\mathrm{i}2\pi rx/L'}=\frac{1}{L'\sigma}\sum_{r\in\mathbb{Z}}e^{-\frac{\pi^2 r^2}{4B'}}e^{\mathrm{i}2\pi rx/L'}
\end{equation}
where $\widetilde{x}=x/L'$ and $\widetilde{\sigma}=\sigma L'$.

\comment{
Its Fourier coefficients are
\begin{equation}
\begin{aligned}
    c_r&=\frac{1}{N}\int_{-N/2}^{N/2}G(x)e^{-\mathrm{i}2\pi rx/N}\mathrm{d}x\\
    &=\frac{1}{N}\int_{-\infty}^{\infty}e^{-\pi\sigma^2x^2}e^{-\mathrm{i}2\pi rx/N}\mathrm{d}x\\
    &=\frac{1}{N\sigma}e^{-\frac{\pi r^2}{N^2\sigma^2}}=\frac{\sqrt{\pi}}{2\sqrt{B}}e^{-\frac{\pi^2r^2}{4B}}.
\end{aligned}
\end{equation}
Therefore
\begin{equation}
    G(x)=\frac{1}{N\sigma}\sum_{r\in\mathbb{Z}}e^{-\frac{\pi r^2}{N^2\sigma^2}}e^{\mathrm{i}2\pi rx/N}=g_x.
\end{equation}
}
Now choose an integer cutoff $K$, and define
\begin{equation}
    P_K(x)=\frac{1}{L'\sigma}\sum_{r=-K}^Ke^{-\frac{\pi^2 r^2}{4B'}}e^{\mathrm{i}2\pi rx/L'}
\end{equation}
and
\begin{equation}
    |\alpha_K\rangle=\frac{1}{\|p\|_2}\sum_{x\in\mathcal{X}_\ell}p_x|x\rangle,
\end{equation}
with $p_x=P_K(x)$.

The errors come from both the periodization and the truncation of the Fourier series. For $|x|\leq L'/2$ and $\nu\neq 0$, we have $|x+\nu L'|\geq |\nu L'|-|x|\geq (|\nu|-1/2)L'$, so
\begin{equation}
    |G(x)-f_x|\leq 2\sum_{\nu=1}^{\infty}e^{-\pi\sigma^2L'^2(\nu-1/2)^2}\leq 2\sum_{\nu=1}^{\infty}e^{-\pi\sigma^2L'^2(\nu-1+1/4)}\leq \frac{2e^{-\pi\sigma^2L'^2/4}}{1-e^{-\pi\sigma^2L'^2}}=\frac{2e^{-B'}}{1-e^{-4B'}}\leq C e^{-B'}.
\end{equation}
For every $x\in\mathbb{R}$, we have
\begin{equation}\label{eq:truncate}
        |G(x)-P_K(x)|\leq \frac{2}{L'\sigma}\sum_{r=K+1}^{\infty}e^{-\frac{\pi r^2}{L'^2\sigma^2}}\leq\frac{2}{L'\sigma}\int_K^\infty e^{-\frac{\pi r^2}{L'^2\sigma^2}}\mathrm{d}r\leq \frac{L'\sigma}{\pi K}e^{-\frac{\pi K^2}{L'^2\sigma^2}}=\frac{2}{\pi^{3/2}}\frac{\sqrt{B'}}{K}e^{-\frac{\pi^2K^2}{4B'}}.
\end{equation}
Therefore
\begin{equation}
    \|f-p\|_2\leq\|f-g\|_2+\|g-p\|_2\leq C\sqrt{N}\left(e^{-B'}+\frac{\sqrt{B'}}{K}e^{-\frac{\pi^2K^2}{4B'}}\right).
\end{equation}
For $|x|<\frac{1}{2\sqrt{\pi}\sigma}$, $f_x\geq e^{-1/4}$, and there are $\Omega(\frac{1}{\sqrt{\pi}\sigma})$ such integer $x$, so
\begin{equation}\label{eq:norm}
    \|f\|_2^2\geq \frac{cL'}{\sqrt{B'}}.
\end{equation}
From the triangle inequality, we can derive
\begin{equation}
\label{eq:triangle}
    \left\|\frac{f}{\|f\|_2}-\frac{p}{\|p\|_2}\right\|_2\leq \frac{2\|f-p\|_2}{\|f\|_2},
\end{equation}
so
\begin{equation}
    \||\alpha_{\ell'}\rangle-|\alpha_K\rangle\|_2\leq CB'^{1/4}\left(e^{-B'}+\frac{\sqrt{B'}}{K}e^{-\frac{\pi^2K^2}{4B'}}\right).
\end{equation}
If we choose $K=\Theta(B')$, then we can achieve state error
\begin{equation}
    \||\alpha\rangle-|\alpha_K\rangle\|_2\leq \||\alpha\rangle-|\alpha_{\ell'}\rangle\|_2+\||\alpha_{\ell'}\rangle-|\alpha_K\rangle\|_2=e^{-\Omega(\ell^2)}.
\end{equation}
There is a consistency condition: we hope the number of Fourier coefficients is smaller than the number of the original Gaussian amplitudes, i.e. $K< 2^{\ell'-1}$, which is satisfied because from Eq.~\eqref{eq:lprimerange}, we have $2^\ell\geq2^{\ell'}>\ell^2$.

Now look at the local unitary preparation circuit depth. We first prepare a state with the Fourier coefficient on $q=\left\lceil\log_2(2K+1)\right\rceil$ qubits, i.e.
\begin{equation}
|\beta_K\rangle_q=\frac{1}{\sqrt S}
\sum_{r=-K}^{K}e^{-\frac{\pi r^2}{L'^2\sigma^2}}|r\rangle_q,~~S=\sum_{r=-K}^Ke^{-\frac{2\pi r^2}{L'^2\sigma^2}}.
\end{equation}
If we write $|\beta_K\rangle_q$ exactly as an MPS, the maximum bond dimension will be $2^{\lfloor q/2\rfloor}$. The sequential preparation of MPS~\cite{PhysRevA.75.032311} from $|0\rangle^{\otimes q}$ requires circuit depth
\begin{equation}
T
\leq
C\sum_{j=1}^{q}4^{1+s_j}
=
\order\!\left(\sum_{j=1}^{q}4^{\min(j,q-j)}\right)
=
\order(2^q)=\order(K),
\end{equation}
where $2^{s_j}\leq\min(2^j,2^{q-j})$ is the bond dimension for the $i$-th qubit, and for the implementation of each isometry, we use the result that an arbitrary $b=1+s_j$ qubit unitary can be implemented using $\order(4^b)$ one-qubit gates and nearest-neighbor \textsf{CNOT}'s, without extra ancillas~\cite{4toq}. 


A better approach in our setting is from Theorem 9 in Ref.~\cite{10.1109/TCAD.2023.3311734}. A direct corollary is that any state on $q$ qubits can be prepared exactly by a nearest-neighbor unitary circuit of depth $\order(2^{q/2}+2^q/\ell)$ from $|0\rangle^{\otimes q}$, using $\ell-q$ ancillas and setting all ancillas to be $|0\rangle$ at the end. In our case, no extra ancillas are needed, because we have $\ell$ qubits in total and we prepare $|\beta_K\rangle$ on only $q=\left\lceil\log_2(2K+1)\right\rceil\ll\ell$ qubits, and the rest qubits can be used as ancillas. In this way, $T=\order(\ell)$ for $K=\Theta(B')=\Theta(\ell^2)$.

Another approach is using the Kitaev-Webb protocol~\cite{Gaussian_Kitaev08} for Gaussian amplitude state preparation, and it gives $T=\order(\text{poly}(q))=\order(\text{polylog}(K))$.

If we represent $r\in\{-2^{q-1},\cdots,2^{q-1}-1\}$ as $|sb_{q-1}\cdots b_1\rangle$, i.e. $r=-s2^{q-1}+\sum_{j=1}^{q-1}b_j2^{j-1}$, $s$ will be the sign qubit, with $s=0$ giving $r>0$ and $s=1$ giving $r<0$. After preparing $|\beta_K\rangle_q$ on $q$ qubits, we need to extend it to $\ell'$ qubits, i.e.
\begin{equation}
    |0\rangle^{\otimes \ell'-q}|r\rangle_q\mapsto |r\rangle_{\ell'},
\end{equation}
which is equivalent to
\begin{equation}
    |0\rangle^{\otimes \ell'-q}|sb_{q-1}\cdots b_1\rangle\mapsto |s\rangle^{\otimes \ell'-q}|sb_{q-1}\cdots b_1\rangle
\end{equation}
in our convention. And this can be done by sequential $\ell'-q$ nearest-neighbor \textsf{CNOT}'s, which produces
\begin{equation}
|B_K\rangle_\ell
=
\frac1{\sqrt S}
\sum_{r=-K}^{K}e^{-\frac{\pi r^2}{L'^2\sigma^2}}|r\rangle_{\ell'}. 
\end{equation}

Then we only need to apply the standard $\ell'$-qubit QFT, which as a nearest-neighbor exact implementation of depth $\order (\ell')$~\cite{PhysRevA.76.052310}, and the only difference is we need to interpret the computational-basis labels by the signed convention we used. Finally, we need to do another extension from $\ell'$ qubits to $\ell$ qubits, i.e.
\begin{equation}
    |0\rangle^{\otimes \ell-\ell'}|r\rangle_{\ell'}\mapsto |r\rangle_{\ell},
\end{equation}
which take depth $\order(\ell-\ell')$.

In total, we can achieve state error
\begin{equation}
    \||\alpha\rangle-|\alpha_K\rangle\|_2=e^{-\Omega(\ell^2)}
\end{equation}
in local unitary preparation circuit depth $T=\order(\ell)+\order(\ell')+\order(\ell-\ell')=\order(\ell)$.

\subsection{The Gaussian amplitude state has small entanglement}
Consider the Taylor expansion of $f_x$ in Eq.~\eqref{eq:gaussianstate} up to a cutoff $m$, i.e. the polynomial
\begin{equation}
    P_m(x)=\sum_{j=0}^m\frac{(-\pi\sigma^2x^2)^j}{j!}.
\end{equation}
Let
\begin{equation}
    |\alpha_m\rangle=\frac{1}{\|p\|_2}\sum_{x\in\mathcal{X}_\ell}p_x|x\rangle,
\end{equation}
with $p_x=P_m(x)$. The pointwise error is bounded by
\begin{equation}
    \max_{x\in\mathcal{X}_\ell}|f_x-p_x|\leq \frac{(\pi\sigma^2x^2)^{m+1}}{(m+1)!}\leq\frac{B^{m+1}}{(m+1)!},
\end{equation}
where we have used the Taylor's theorem with the Lagrange form of the remainder. Because $|\mathcal{X}_\ell|=L$, we have
\begin{equation}
    \|f-p\|_2\leq\sqrt{L}\frac{B^{m+1}}{(m+1)!}.
\end{equation}
Using Eq.~\eqref{eq:triangle}, we have the state error
\begin{equation}
    \left\||\alpha\rangle-|\alpha_m\rangle\right\|_2\leq\frac{2e^BB^{m+1}}{(m+1)!}\leq 2e^B\left(\frac{eB}{m+1}\right)^{m+1}\leq\epsilon,
\end{equation}
where we have used $\|f\|_2\geq\sqrt{L}e^{-B}$ and the Stirling's formula. A sufficient condition for $m$ to upper bound the state error by $\epsilon$ is
\begin{equation}
    m+1=\left\lceil\frac{4(B+\Lambda)}{\log{(2+\frac{\Lambda}{B})}}\right\rceil,
\end{equation}
where $\Lambda=\log\frac{2}{\epsilon}$, because
\begin{equation}
\begin{aligned}
    (m+1)\log\left(\frac{m+1}{eB}\right)&\geq \frac{4(B+\Lambda)}{\log{(2+\frac{\Lambda}{B})}}\log\frac{4(1+\Lambda/B)}{e\log{(2+\frac{\Lambda}{B})}}\\
    &\geq\frac{4(B+\Lambda)}{\log{(2+\frac{\Lambda}{B})}}\log\frac{4(1+\Lambda/B)}{e\sqrt{1+\Lambda/B}}=\frac{4(B+\Lambda)}{\log{(2+\frac{\Lambda}{B})}}\log\frac{4\sqrt{1+\Lambda/B}}{e}\\
    &\geq\frac{4(B+\Lambda)}{\log{(2+\frac{\Lambda}{B})}}\log\sqrt{2+\Lambda/B}=2(B+\Lambda).
\end{aligned}
\end{equation}

From Proposition 2.8 in Ref.~\cite{Khoromskij2011}, we know that an equidistantly sampled polynomial of degree $p$ has matrix product state (MPS) bond dimension at most $p+1$. In our case $p=2m$. Therefore, if we want to represent $|\alpha\rangle$ by a MPS with error upper bounded by $\epsilon$, we need bond dimension at most
\begin{equation}
    \chi=\order(2m+1)=\order\left(2\left\lceil\frac{4(B+\log\frac{2}{\epsilon})}{\log{(2+\frac{\log\frac{2}{\epsilon}}{B})}}\right\rceil-1\right).
\end{equation}
We list some example upper bounds for MPS bond dimension $\chi$ under difference choices of $\epsilon$ and $B$ as below. In other words, $|\alpha\rangle$ can be exponentially well approximated by $|\alpha_m\rangle$ which has entanglement entropy only $\order(\log\ell)$.
\begin{table}[h!]
    \centering
    \begin{tabular}{c|ccc}
    \hline
        $B$ & $\epsilon=\Theta(1)$ & ~~$\epsilon=e^{-\Theta(\ell)}$~~ & ~~$\epsilon=e^{-\Theta(\ell^2)}$~~ \\
        \hline
        $\order(1)$ & $\order(1)$ &$\order(\ell/\log\ell)$&$\order(\ell^2/\log\ell)$\\
        $\order(\ell)$&$O(\ell)$&$\order(\ell)$&$\order(\ell^2/\log\ell)$\\
        $\order(\ell^2)$&$\order(\ell^2)$&$\order(\ell^2)$&$\order(\ell^2)$\\
    \hline
    \end{tabular}
\end{table}

This upper bound of the bond dimension given precision $e^{-\Omega(\ell^2)}$ shown in the table above for $B=o(\ell^2)$ generally implies local unitary preparation circuit depth $\order{(\mathrm{poly}(\ell))}$ by sequential preparation of MPS. Whether or not the explicit Pascal form of the MPS representation we will show later can be used to shorten the circuit depth remains an open question.

For $\Theta(\ell^2)\leq B<\Theta(2^{2\ell}/\ell^2)$, the Gaussian is effectively supported on a smaller system and we can define an effective $B'=\Theta(\ell^2)$ like in the last section, and therefore it still has a small entanglement $O(\log\ell)$ with the same state precision. For $B=\Omega(2^{2\ell}/\ell^2)$, the Gaussian is even narrower and closer to a product state. Generally for any $B\geq 1$, we have the bond dimension
\begin{equation}
\chi
=
\order\left[
\min\left\{
1+
\frac{B_*+\Lambda}
{\log(2+\Lambda/B_*)},
\;
1+\left(\frac{2^{2\ell}\Lambda}{B}\right)^{1/4}
\right\}
\right],
\end{equation}
where
\begin{equation}
\Lambda=\max\!\left\{1,\log\frac{c}{\epsilon}\right\},
\qquad
B_*=\min\{B,\Lambda\},
\end{equation}
with $c$ a constant.

$|\alpha_m\rangle$ actually has an explicit MPS representation given by the Pascal matrices. Theorem 6 in Ref.~\cite{Oseledets2013} gives
\begin{equation}
P_m(x)=\bm{c}^\top A^{s_1}A^{s_2}\cdots A^{s_\ell}\bm{e}_0,
\end{equation}
where $s_j=\{0,1\}$,
\begin{equation}
    x=\sum_{j=1}^{\ell}2^{j-1}s_j-2^{\ell-1},
\end{equation}
$\bm{e}_0=(1,0,\cdots,0)^\top$, and $\bm{c}=(c_0,\cdots,c_{2m})^\top$ is defined as the coefficient vector of $P_m(x)$, i.e.
\begin{equation}
    P_m(x)=\bm{c}^\top\cdot \bm{v}(x)
\end{equation}
with
\begin{equation}
    c_{2j}=\frac{(-\pi\sigma^2)^j}{j!},~~c_{2j+1}=0,
\end{equation}
and 
\begin{equation}
    \bm{v}(x)=(1,x,x^2,\cdots,x^{2m})^\top
\end{equation}
the monomial vector. The bulk MPS tensor is the Pascal matrix 
\begin{equation}
    A^{s_j}_{ab}=\left\{
    \begin{aligned}
        &\left(\begin{matrix}
            a\\b
        \end{matrix}\right)[t_j(s_j)]^{a-b},& a\geq b\\
        &0,&a<b
    \end{aligned}
    \right.
\end{equation}
which translates the monomial vector, i.e. $A(t)\bm{v}(x)=\bm{v}(x+t)$, with
\begin{equation}
    t_j(b_j)=2^{j-1}b_j,\text{ for }0\leq j<\ell,\text{ and }t_\ell(b_\ell)=2^{\ell-1}(b_\ell-1).
\end{equation}

\section{Entanglement entropy of spin squeezed states: Proof of Theorem \ref{thm:bound}}

The sketch of the proof is as follows. First, we truncate the Dicke-state expansion of $|\psi\rangle$ around $\langle Z\rangle_\psi$ to obtain an unnormalized truncated $\ket{\psi^{\rm trun}}$. Next, we bound the largest eigenvalue, equivalently the operator norm, of the unnormalized reduced density matrix
\begin{equation}
    \rho_L^{\mathrm{trun}}=\mathrm{Tr}_R|\psi^{\mathrm{trun}}\rangle\langle\psi^{\mathrm{trun}}|    
\end{equation}
of the left subsystem by bounding the largest Schmidt coefficient of each retained Dicke state and then using the corresponding amplitudes to bound their combined contribution. Finally, the support of any approximate unnormalized density matrix $\epsilon$-close to
\begin{equation}
    \rho_L=\mathrm{Tr}_R|\psi\rangle\langle\psi|
\end{equation}
must retain at least a certain amount of the truncated state $\rho_L^{\mathrm{trun}}$’s weight. Since each dimension of that support can contribute at most the operator norm of $\rho_L^{\mathrm{trun}}$, capturing this weight requires sufficiently large rank, yielding the desired smooth max-entropy bound.

The lower bound on the depth of one-dimensional local adaptive circuits follows because each layer can increase the rank of the reduced density matrix by at most a constant factor on each measurement branch.

In this section, we assume that $\ket{\psi}$ lies in the
permutation-symmetric Dicke manifold and satisfies
\begin{equation}
|\langle X\rangle_\psi|\ge\kappa N,\qquad \operatorname{Var}_{\psi}(Z)=\langle\Delta Z^2\rangle_{\psi}\le N\xi^2,
\end{equation}
    where $0<\kappa\le1,0<\xi\le1$. Since $\langle X\rangle_\psi^2+\langle Z\rangle_\psi^2\le N^2$, we have \begin{equation}
    z_0:=\langle Z\rangle_\psi, \qquad \left|z_0\right|\le(1-2\zeta)N,
\end{equation}
where $\zeta=\frac{1-\sqrt{1-\kappa^2}}{2}$.

\subsection{Truncating the state to polarizations near $\langle Z\rangle_\psi$}
We truncate $\ket{\psi}$ in the Dicke-state basis and obtain an unnormalized approximation \begin{equation}\label{eq:psi_trun}
    \ket{\psi^{\rm trun}}:=\sum_{|z-z_0|<w} \beta_z \ket{D_z},  
\end{equation}
where $\beta_z = \alr{D_z|\psi}$ and $\ket{D_z}=\ket{D^N_{(N-z)/2}}$ for \(z\in\{-N,-N+2,\ldots,N\}\). In other words, $\psi^{\mathrm{trun}}$ is obtained by truncating
$\psi$ to the subspace $|z-z_0|<w$. To control the discarded
weight, we introduce a parameter $\eta\in(0,\sqrt{1-\epsilon})$ independent of $N$
and choose the half-width of the truncation window as
\begin{equation}\label{eq:n'=}
    w:=\frac{\xi\sqrt{N}}{\eta}.
\end{equation}
At sufficiently large $N$, we can assume that \begin{equation}
    w\le\zeta N, \label{eq:width}
\end{equation}
so that the truncation window remains away from the fully polarized endpoints.

Define the discarded component by
\begin{equation}
|\psi_\perp\rangle
:=|\psi\rangle-|\psi^{\mathrm{trun}}\rangle
=\sum_{|z-z_0|\ge w}\beta_z|D_z\rangle,
\end{equation}
which is orthogonal to $|\psi^{\mathrm{trun}}\rangle$.
Chebyshev's inequality then gives the discarded weight during the truncation
\begin{equation}
\begin{aligned}
\|\psi_\perp\|^2
&=\sum_{|z-z_0|\ge w}|\beta_z|^2
\le\frac{1}{w^2}\sum_z(z-z_0)^2|\beta_z|^2=\frac{\langle\psi|(Z-z_0)^2|\psi\rangle}{w^2}
=\frac{\operatorname{Var}_\psi(Z)}{w^2}
\le\frac{N\xi^2}{w^2}
=\eta^2.
\end{aligned}
\end{equation}

\subsection{Bounding the largest eigenvalue of the truncated reduced density matrix} \label{sec:opnorm}

Our goal is to bound the operator norm $\lVert\rho_L^{\mathrm{trun}}\rVert_{\mathrm{op}}$, i.e. the largest eigenvalue of $\rho_L^{\mathrm{trun}}$. We first bound the largest Schmidt coefficient of each
retained Dicke state. We then combine these bounds for
the truncated superposition to obtain an upper bound on
$\|\rho_L^{\mathrm{trun}}\|_{\mathrm{op}}$.

Across the bipartition into $N/2$ spins on each side, write
the Dicke state with total $Z$-magnetization $z$ as
\begin{equation}
|D_z\rangle=\sum_m C_{m,z/2-m}|m\rangle_L\otimes|z/2-m\rangle_R,
\end{equation}
where $\ket{m}_L$ is the permutation-symmetric state of the left
subsystem with total $Z$-magnetization $2m$, with
$m=-N/4,-N/4+1,\ldots,N/4$. The amplitudes $C_{m,z/2-m}$
are Clebsch--Gordan (CG) coefficients.

The squared magnitude of the CG coefficient is
\begin{equation}\label{eq:CG=}
p_z(m):=|C_{m,z/2-m}|^2
=
\frac{
    \dbinom{N/2}{N/4-m}\,
    \dbinom{N/2}{N/4-z/2+m}
}{
    \dbinom{N}{N/2-z/2}
}.
\end{equation}

We look for the largest CG coefficient by evaluating the ratio of two consecutive terms
\begin{equation}
\frac{p_z(m+1)}{p_z(m)}=\frac{(N/4-m)(N/4+z/2-m)}{(N/4+m+1)(N/4-z/2+m+1)}.
\end{equation}

The numerator of the right-hand side minus its denominator is
\begin{equation}
-\left(\frac N2+1\right)(2m+1-z/2).
\label{eq:25}
\end{equation}
Thus the sequence $p_z(m)$ increases before $m\simeq z/4$ and decreases after $m\simeq z/4$. The maximum of $p_z(m)$ is therefore reached at $m_*$ which has the form
\begin{equation}
m_*=\frac z4+s,\qquad |s|\le1.
\label{eq:26}
\end{equation}
The bounded offset $s$ accommodates parity and rounding. At this point, we have
\begin{equation}
p_z(m_*)
=
\frac{
    \dbinom{N/2}{(N-z)/4-s}\,
    \dbinom{N/2}{(N-z)/4+s}
}{
    \dbinom{N}{(N-z)/2}
}.
\end{equation}

Using $|z_0|\le(1-2\zeta)N$ and $w\le\zeta N$,
every retained label $z$ satisfies
\begin{equation}
|z|\le |z_0|+|z-z_0|
<(1-2\zeta)N+w
\le(1-\zeta)N.
\end{equation}
Consequently, $(2-\zeta)N/4\geq(N\pm z)/4\ge\zeta N/4$.
Together with $|s|\le1$, this ensures that all factorial
arguments grow proportionally to $N$. We can therefore
apply Stirling's formula as $N\to\infty$, with relative
corrections bounded by a constant times $1/N$, where the constant may depend on $\zeta$ but is independent
of $z$:
\begin{equation}\label{eq:CG-num-stirling}
\dbinom{N/2}{(N-z)/4-s}\,
\dbinom{N/2}{(N-z)/4+s}
=
\frac{
    4\cdot 2^{N h_2((N-z)/(2N))}
}{
    \pi N\left[1-(z/N)^2\right]
}
\left(1+O(N^{-1})\right),
\end{equation}
\begin{equation}\label{eq:CG-den-stirling}
\dbinom{N}{(N-z)/2}
=
\frac{
    2^{N h_2((N-z)/(2N))}
}{
    \sqrt{\frac{\pi N}{2}\left[1-(z/N)^2\right]}
}
\left(1+O(N^{-1})\right),
\end{equation}
where $h_2(x)=-x\log_2 x-(1-x)\log_2(1-x)$ is the
binary entropy. Rewriting the exact expressions in terms of central
binomial coefficients shows that the multiplicative
correction factor relative to the leading term in
Eq.~(\ref{eq:CG-num-stirling}) is smaller than the
corresponding factor in Eq.~(\ref{eq:CG-den-stirling}).
Upon division, the ratio of these correction factors
is therefore less than one and can be dropped to obtain
the upper bound
\begin{equation}
\max_m p_z(m)
<
\sqrt{\frac{8}{\pi N[1-(z/N)^2]}}
\le
\sqrt{\frac{8}{\pi N\zeta(2-\zeta)}}.\label{eq:maxm}
\end{equation}

We now relate the individual Dicke-state coefficient bounds
to the operator norm of $\rho_L^{\mathrm{trun}}$.
For each $z$, define the coefficient matrix $\Psi^{(z)}$
of $|D_z\rangle$ in the left and right Dicke bases by
\begin{equation}
\Psi^{(z)}_{m,m'}
:=C_{m,m'}\delta_{m+m',z/2},
\end{equation}
where $\delta$ is the Kronecker delta.
The coefficient matrix of $|\psi^{\mathrm{trun}}\rangle$
is then
\begin{equation}
\Psi:=\sum_{|z-z_0|<w}\beta_z\Psi^{(z)},
\end{equation}
so that
\begin{equation}
|\psi^{\mathrm{trun}}\rangle
=
\sum_{m,m'}\Psi_{m,m'}|m\rangle_L\otimes|m'\rangle_R.
\end{equation}
Tracing out the right subsystem gives, in the left Dicke basis,
\begin{equation}
\rho_L^{\mathrm{trun}}=\Psi\Psi^\dagger,
\qquad
\|\rho_L^{\mathrm{trun}}\|_{\mathrm{op}}
=\|\Psi\|_{\mathrm{op}}^2.
\end{equation}

For each fixed $z$, the Kronecker delta fixes
$m'=z/2-m$, so $\Psi^{(z)}$ has at most one nonzero
entry in each row and column. Consequently,
$\|\Psi^{(z)}\|_{\mathrm{op}}=\max_m|C_{m,z/2-m}|$.
We therefore obtain
\begin{align}
\|\Psi\|_{\mathrm{op}}
&\le
\sum_{|z-z_0|<w}|\beta_z|
\|\Psi^{(z)}\|_{\mathrm{op}}
\nonumber\\
&=
\sum_{|z-z_0|<w}|\beta_z|
\max_m|C_{m,z/2-m}|
\nonumber\\
&\le
\left(\sum_{|z-z_0|<w}|\beta_z|^2\right)^{1/2}
\left(\sum_{|z-z_0|<w}
\max_m|C_{m,z/2-m}|^2\right)^{1/2}
\nonumber\\
&\le
\left(\sum_{|z-z_0|<w}
\max_m|C_{m,z/2-m}|^2\right)^{1/2}
\nonumber\\
&\le
\sqrt{w+1}\,
\max_{\substack{|z-z_0|<w\\m}}|C_{m,z/2-m}|
\nonumber\\
&=
\sqrt{(w+1)\max_{\substack{|z-z_0|<w\\m}}p_z(m)}.
\label{eq:Psi-pz-bound}
\end{align}
Here the first inequality follows from the triangle inequality.
Applying the Cauchy--Schwarz inequality and using
$\sum_{|z-z_0|<w}|\beta_z|^2\le1$ then gives the fourth line.
We further bound the sum by its largest term times the
number of retained labels to obtain
the second-to-last line; there are at most $w+1$
such labels because consecutive allowed values of $z$
differ by two.
The last equality uses the definition of $p_z(m)$.

Substituting $w=\xi\sqrt{N}/\eta$ and
Eq.~(\ref{eq:maxm}), we obtain
\begin{equation}
\|\rho_L^{\mathrm{trun}}\|_{\mathrm{op}}
=\|\Psi\|_{\mathrm{op}}^2
\le
(w+1)\max_{\substack{|z-z_0|<w\\m}}p_z(m)
\le
(w+1)\sqrt{\frac{8}{\pi N\zeta(2-\zeta)}}
\le
\sqrt{\frac{9}{\pi\zeta(2-\zeta)}}
\frac{\xi}{\eta}
\label{eq:rho-trun-op-bound}
\end{equation}
where we have used $N$ is sufficiently large. Recall that $\eta$ is the parameter we used earlier to control the discarded weight in the Dicke-state truncation.

\subsection{Converting the eigenvalue bound into smooth rank}
We show that the support of every admissible approximation must capture some of the truncated weight. Since each support dimension contributes at most the largest
eigenvalue bounded in Sec.~\ref{sec:opnorm}, this forces a lower bound on the support dimension. We adapt this proof from Proposition 6.3 in \cite{tomamichel_book}.

Take an arbitrary operator $\widetilde\rho_L\ge0$ satisfying
\begin{equation}
\lVert\widetilde\rho_L-\rho_L\rVert_1\le\epsilon.
\label{eq:71}
\end{equation}
Let $Q$ be the orthogonal projector onto
$\operatorname{supp}(\widetilde\rho_L)$, acting on the
left subsystem. Then
\begin{equation}
\operatorname{rank}(\widetilde\rho_L)=\operatorname{Tr}Q.
\end{equation}
To lower-bound this rank, we first show that the
approximation condition forces the subspace onto which
$Q$ projects to retain at least $1-\epsilon$ of the
original state's weight:
\begin{equation}
\label{eq:Q-psi-weight}
\begin{aligned}
\|Q|\psi\rangle\|^2
&=\operatorname{Tr}(Q\rho_L)\\
&=1-\operatorname{Tr}\bigl((I-Q)
(\rho_L-\widetilde\rho_L)\bigr)\\
&\ge
1-\left|\operatorname{Tr}\bigl((I-Q)
(\rho_L-\widetilde\rho_L)\bigr)\right|\\
&\ge
1-\|I-Q\|_{\mathrm{op}}\,
\|\rho_L-\widetilde\rho_L\|_1\\
&\ge1-\epsilon.
\end{aligned}    
\end{equation}
Here $(I-Q)\widetilde\rho_L=0$ and $\|I-Q\|_{\mathrm{op}}\le1$ are used.

Using the bound on the discarded component, we next
obtain a lower bound on the norm of the truncated state
projected onto the same subspace:
\begin{equation}
\begin{aligned}
\lVert Q|\psi^{\mathrm{trun}}\rangle\rVert
&=\lVert Q(|\psi\rangle-|\psi_\perp\rangle)\rVert\\
&\ge\lVert Q|\psi\rangle\rVert-\lVert Q|\psi_\perp\rangle\rVert\\
&\ge\sqrt{1-\epsilon}-\lVert Q\rVert_{\mathrm{op}}\,\lVert\psi_\perp\rVert\\
&\ge\sqrt{1-\epsilon}-\lVert\psi_\perp\rVert\\
&\ge\sqrt{1-\epsilon}-\eta.
\end{aligned}    
\end{equation}
For \begin{equation}
    \operatorname{Tr}(Q\rho_L^{\mathrm{trun}}) = \lVert Q|\psi^{\mathrm{trun}}\rangle\rVert^2,
\end{equation}
each orthogonal direction in
$\operatorname{supp}(\widetilde\rho_L)$ contributes at most
$\|\rho_L^{\mathrm{trun}}\|_{\mathrm{op}}$ to the weight
$\operatorname{Tr}(Q\rho_L^{\mathrm{trun}})$.
The support must therefore have sufficiently many
dimensions to accommodate this weight:
\begin{equation}
\operatorname{rank}(\widetilde\rho_L)
\ge
\frac{\operatorname{Tr}(Q\rho_L^{\mathrm{trun}})}
{\|\rho_L^{\mathrm{trun}}\|_{\mathrm{op}}}
\ge
\frac{(\sqrt{1-\epsilon}-\eta)^2}
{\displaystyle
\sqrt{\frac{9}{\pi\zeta(2-\zeta)}}
\frac{\xi}{\eta}}
\ge \frac{1}{C_{\epsilon,\kappa}\xi}
\end{equation}
for a constant $C_{\epsilon,\kappa}$ determined by $\epsilon,\kappa$. Here we have chosen for example $\eta=\sqrt{(1-\epsilon)/6}$.

\comment{
The last inequality uses $N^{-1/2}\le\xi/\kappa$, where
\begin{equation}
C_{\epsilon,\kappa}
:=
\sqrt{\frac{32}
{\pi\left(2+\kappa^2-2\sqrt{1-\kappa^2}\right)}}
\frac{\eta^{-1}+\kappa^{-1}}
{(\sqrt{1-\epsilon}-\eta)^2}.
\end{equation}
The parameter $\eta$ can be
chosen from $(0,\sqrt{1-\epsilon})$ to minimize $C_{\epsilon,\kappa}$. Choosing $\eta$ too large increases the truncation error,
while choosing it too small makes the bound on
$\|\rho_L^{\mathrm{trun}}\|_{\mathrm{op}}$ too weak to be useful.

As an example, for convenience, we can choose
\begin{equation}
\eta=\frac{\sqrt{1-\epsilon}}{6},
\end{equation}
which gives
\begin{equation}
C_{\epsilon,\kappa}
:=
\frac{36}{25(1-\epsilon)}
\sqrt{\frac{32}
{\pi\left(2+\kappa^2-2\sqrt{1-\kappa^2}\right)}}
\left(\frac{6}{\sqrt{1-\epsilon}}+\frac{1}{\kappa}\right).
\end{equation}
}

Therefore, from the definition of the smooth max-entropy, we have
\begin{equation}
H^{\epsilon}_{\rm max}(\rho_L)
=
\min_{\substack{\widetilde{\rho}_L\ge0\\
\norm{\widetilde{\rho}_L-\rho_L}_1\le\epsilon}}
H_{\rm max}(\widetilde{\rho}_L)
\ge
\log\frac{1}{C_{\epsilon,\kappa}\xi}.
\end{equation}

\subsection{Entanglement bounds circuit depth}

Finally, we rigorously prove the circuit-depth lower bound
in Theorem~\ref{thm:bound}.
Suppose a 1d adaptive circuit with measurements and nonlocal
classical feedforward prepares an ensemble of states
$\{\ket{\varphi_m}\}$ labeled by the measurement outcomes $m$.
The circuit starts from a product state and uses quantum
operations of bounded size and range.
If some nonzero-probability outcome $m$ gives a state close
to a $\xi$-SSS $\ket{\psi}$ satisfying
$|\langle X\rangle_\psi|\ge\kappa N$ and
\begin{equation}\label{eq:prepare_close}
\norm{\ket{\varphi_m}\bra{\varphi_m}
-\ket{\psi}\bra{\psi}}_1\le\epsilon,
\end{equation}
then the circuit requires depth
$\Omega\lr{\log\frac{1}{\xi}}$ for fixed $\epsilon<1$
and $\kappa>0$.
The bound also holds for protocols with postselection.

\begin{proof}
Fix such a measurement outcome $m$ and write
$\ket{\varphi}=\ket{\varphi_m}$.
Since partial trace cannot increase the trace-norm distance,
\begin{equation}\label{eq:rho-rho<psi-psi}
\norm{\widetilde{\rho}_L-\rho_L}_1
\le
\norm{\ket{\varphi}\bra{\varphi}
-\ket{\psi}\bra{\psi}}_1
\le\epsilon,
\end{equation}
where
$\widetilde{\rho}_L=\trace_R\lr{\ket{\varphi}\bra{\varphi}}$.
Thus $\widetilde{\rho}_L$ satisfies the constraints in the minimization
defining the smooth max-entropy.
Applying Eq.~(\ref{eq:maxentropy>}) gives
\begin{equation}\label{eq:Hmaxphi>}
H_{\rm max}(\widetilde{\rho}_L)
\ge H^\epsilon_{\rm max}(\rho_L)
\ge\log\frac{1}{C_{\epsilon,\kappa}\xi}.
\end{equation}

To relate this entropy bound to circuit depth, note that
$\operatorname{rank}(\widetilde{\rho}_L)$ equals the
Schmidt rank of $\ket{\varphi}$ across the middle cut.
Fixing the measurement record also fixes every adaptive
choice, so the branch is a definite sequence of local
Kraus operators.
Operations supported entirely on either side of the cut
cannot increase Schmidt rank.
This includes conditioning on a particular outcome of a
measurement within either half.
The Schmidt rank can increase only when a quantum operation
acts on qubits on both sides of the cut.
Since each operation is local,
it can multiply the Schmidt rank by at most a constant factor.
Moreover, in one dimension, the locality of the operations
allows only a constant number of them to cross the cut in
each layer. Thus each circuit layer can increase the Schmidt
rank by at most a constant factor.
Since the input product state has Schmidt rank one, reaching
the Schmidt rank of $\ket{\varphi}$ requires a circuit
depth $D$ satisfying
\begin{equation}
\operatorname{rank}(\widetilde{\rho}_L)
\le\ee^{\order(D)},
\qquad
H_{\rm max}(\widetilde{\rho}_L)=\order(D).
\end{equation}
This rank bound holds on every measurement branch and
therefore remains valid under postselection.
Combining it with Eq.~(\ref{eq:Hmaxphi>}) gives
$D=\Omega(\log(1/\xi))$ for fixed $\epsilon<1$
and $\kappa>0$.
\end{proof}

The proof above relies on the 1d geometry so that the two subsystems $L,R$ only interact via $\order(1)$ gates per layer. We leave as an open question whether $\xi$-SSSs can be prepared in much shorter depth in higher dimensions.

\section{A lower bound for unitary circuit preparation}
If a state $|\psi\rangle$ on $N$ qubits is prepared by a depth-$L$ local unitary circuit from a product state, then we have
\begin{equation}
\label{eq:locc}
    \langle \psi |O_AO_B|\psi\rangle=\langle\psi | O_A|\psi\rangle\langle\psi | O_B|\psi\rangle,~~\forall A,B\subset\Lambda \text{ with }d(A,B)>2L,
\end{equation}
where $A,B$ are two regions in the lattice $\Lambda$ whose distance is defined as $d(A,B)=\min_{i\in A,j\in B}d(i,j)$, with $d(i,j)$ the minimum number of edges between the vertices $i$ and $j$ in the graph of the lattice $\Lambda$~\cite{adaptive_phase21}. Define
\begin{equation}
    b(2L)=\max_i|\{j:d(i,j)\leq 2L\}|,
\end{equation}
where $|\cdot |$ denotes the cardinality of the set. Let $Y=\sum_iY_i$. Let $O_A=Y_i$ and $O_B=Y_j$, we have
\begin{equation}
    \text{Cov}(Y_i,Y_j)=0,~~\forall d(i,j)\geq 2L.
\end{equation}
Using the Cauchy-Schwarz inequality, we have
\begin{equation}
    |\text{Cov}(Y_i,Y_j)|\leq\sqrt{\text{Var}(Y_i)\text{Var}{Y_j}}\leq 1,
\end{equation}
because $Y_i$ is a Pauli operator. Therefore
\begin{equation}
    \text{Var}(Y)=\sum_{i,j}\text{Cov}(Y_i,Y_j)\leq Nb(2L).
\end{equation}
Let $|\psi\rangle$ be a SSS. The uncertainty principle yields
\begin{equation}
    \text{Var}(Y)\text{Var}(Z)\geq \frac{1}{4}|\langle [Y,Z]\rangle|^2=\frac{1}{4}|2\mathrm{i}\langle X\rangle|^2=|\langle X\rangle|^2=cN^2,
\end{equation}
with $c>0$ a constant, where we have used Eq.~\eqref{eq:X=N} from the definition of a SSS. Combining with Eq.~\eqref{eq:DZ2=xi}, we have
\begin{equation}
    \xi^2\geq \frac{\text{Var}(Z)}{N}\geq \frac{cN^2}{N\text{Var}(Y)}\geq\frac{c}{b(2L)}.
\end{equation}
Thus we proved that a depth-$L$ local unitary circuit from a CSS can prepare a SSS with at most $1/\xi^2=b(2L)/c$. Or equivalently, it is a lower bound for the local unitary circuit depth to prepare a $\xi$-SSS.

Now we restrict the SSS to be within the Dicke manifold. Denote the mean collective spin of the state as $\langle \textbf{J}\rangle$, we can always rotate the coordinates so that $\langle \textbf{J}\rangle$ points along the $x$-direction. Then by definition we have $\langle Y\rangle=\langle Z\rangle=0$. Permutation symmetry further gives
\begin{equation}
    \langle Z_i\rangle=0,~~\langle Z_iZ_j\rangle=C,~~\forall i\neq j.
\end{equation}
Therefore combining with Eq.~\eqref{eq:DZ2=xi}, we have
\begin{equation}
    N\xi^2\geq\text{Var}(Z)=N+N(N-1)C,
\end{equation}
which implies
\begin{equation}
    C\leq\frac{\xi^2-1}{N-1}<0,~~\forall \xi^2<1.
\end{equation}
As a result, connected correlations at any distance is nonzero and Eq.~\eqref{eq:locc} is violated for any constant finite $L$. Thus we proved that any $\xi$-SSS in the Dicke manifold with $\xi<1$ cannot be prepared by finite depth local unitary circuit. Notice that this is not true if we don't impose the condition that the SSS is within the Dicke manifold. We can consider the following simple example. Arrange $N$ qubits on a 1d chain and pair nearest neighbors as
\begin{equation}
    |\psi\rangle=\bigotimes_{j=1}^{N/2}|\phi\rangle_{2j-1,2j},\text{ with }|\phi\rangle=\cos{\left(\frac{\pi}{12}\right)}|++\rangle-\sin{\left(\frac{\pi}{12}\right)|--\rangle}.
\end{equation}
For a single pair, direct calculation gives
\begin{equation}
\begin{aligned}
    &\langle\phi|X_1|\phi\rangle=\langle\phi|X_2|\phi\rangle=\cos{\left(\frac{\pi}{6}\right)}=\frac{\sqrt{3}}{2},\\
    &\langle\phi|Z_1|\phi\rangle=\langle\phi|Z_2|\phi\rangle=\langle\phi|Y_1|\phi\rangle=\langle\phi|Y_2|\phi\rangle=0,\\
    &\langle\phi|Z_1Z_2|\phi\rangle=-\langle\phi|Y_1Y_2|\phi\rangle=-\sin{\left(\frac{\pi}{6}\right)}=-\frac{1}{2}.
\end{aligned}
\end{equation}
Therefore
\begin{equation}
\begin{aligned}
    \text{Var}_\phi(Z_1+Z_2)=\langle\phi|Z_1^2+Z_2^2+2Z_1Z_2|\phi\rangle=1\\
    \text{Var}_\phi(Y_1+Y_2)=\langle\phi|Y_1^2+Y_2^2+2Y_1Y_2|\phi\rangle=3.
\end{aligned}
\end{equation}
Because $|\psi\rangle$ is a product state of $|\phi\rangle$, the means and variances add. We have
\begin{equation}
    \langle\psi| X|\psi\rangle=\frac{\sqrt{3}}{2}N,~~\langle\psi| Y|\psi\rangle=\langle\psi| Z|\psi\rangle=0,~~\text{Var}_{\psi}(Z)=\frac{N}{2},~~\text{Var}_{\psi}(Y)=\frac{3N}{2}.
\end{equation}
This is a $\xi$-SSS with $\xi=1/\sqrt{2}$ but can be prepared by a depth-1 nearest-neighbor unitary circuit from a CSS.

\section{More about the simple example of SSS}
Consider the normalized state
\begin{equation}
    |\psi\rangle=\frac{|\alpha\alpha\cdots\alpha\rangle+|\beta\beta\cdots\beta\rangle}{\sqrt{2+2(\text{Re}\langle\alpha |\beta\rangle)^N}},
\end{equation}
where $|\alpha\rangle$ and $|\beta\rangle$ are two CSS on the Bloch sphere differ by angle $2\theta\sim1/\sqrt{N}$. We have
\begin{equation}
    \text{Re}\langle\alpha |\beta\rangle=\cos\theta\approx 1-\frac{\theta^2}{2}.
\end{equation}
Therefore
\begin{equation}
    (\text{Re}\langle\alpha |\beta\rangle)^N\approx\exp{\left(-\frac{N\theta^2}{2}\right)}.
\end{equation}
For $\theta\sim 1/\sqrt{N}$, we have
\begin{equation}
    (\text{Re}\langle\alpha |\beta\rangle)^N\sim e^{-1/2}.
\end{equation}
So the two branches are not orthogonal for arbitrarily large $N$. Thus it is not a macroscopic GHZ cat state.

We can show that $|\psi\rangle$ is a SSS. Without loss of generality, we can put $|\alpha\rangle$ and $|\beta\rangle$ at the equator of the Bloch sphere and let them located symmetrically about the $x$-axis. Then we have
\begin{equation}
\begin{aligned}
     \langle\psi|X|\psi\rangle=\frac{N(\cos\theta+\cos^{N-1}\theta)}{1+\cos^N\theta}\sim N,&~~\langle\psi|Y|\psi\rangle=\langle\psi|Z|\psi\rangle=0,\\
     \text{Var}_\psi (Z)=N-N(N-1)\frac{\sin^2\theta\cos^{N-2}\theta}{1+\cos^N\theta},&~~\text{Var}_\psi(Y)=N+N(N-1)\frac{\sin^2\theta}{1+\cos^N\theta}.
\end{aligned}
\end{equation}
Therefore
\begin{equation}
    \xi^2=1-(N-1)\frac{\sin^2\theta\cos^{N-2}\theta}{1+\cos^N\theta}<1
\end{equation}
as long as
\begin{equation}
    N\sin^2\theta-1<\frac{1}{\cos^{N-2}\theta},
\end{equation}
which is true for small enough $\theta\sim 1/\sqrt{N}$.

It is always possible to find an operator $O$ that distinguishes $|\alpha\rangle$ and $|\beta\rangle$, i.e. 
\begin{equation}
    \langle\alpha |O|\alpha\rangle=-\langle\beta |O|\beta\rangle=\sin{\theta},\text{ and }\langle\alpha |O|\beta\rangle=0.
\end{equation}
And in the example above, $O_i=Y_i$. We have
\begin{equation}
\begin{aligned}
    \langle\psi | O_i|\psi\rangle&=0,\\
    \langle\psi | O_iO_j|\psi\rangle&=\frac{\sin^2\theta}{1+\cos^N\theta}.
\end{aligned}
\end{equation}
The connected correlation then becomes
\begin{equation}
    \langle\psi |O_iO_j|\psi\rangle-\langle\psi |O_i|\psi\rangle\langle\psi| O_j|\psi\rangle=\langle\psi |O_iO_j|\psi\rangle\sim\frac{1}{N}\neq 0.
\end{equation}
Therefore $|\psi\rangle$ has non-zero correlations at arbitrarily large spatial separation and is long-range correlated, though the correlation vanishes as $N\rightarrow \infty$ and $|\psi\rangle$ becomes locally indistinguishable from a product state in the thermodynamic limit. Those features are similar to the $W$ state. Consequently, $|\psi\rangle$ cannot be prepared by finite depth local unitary circuit.

\section{Numerical simulation of the protocol under incoherent noises}
\subsection{Filter parameter settings}
Ideally, our protocol first prepares the ancillae $A$ to the state
\begin{equation}
\label{eq:sum}
\sum_{x=-2^{\ell-1}}^{2^{\ell-1}-1}a_x\ket{x},
\end{equation}
with the unnormalized Gaussian amplitudes
\begin{equation}
    a_x=\frac{1}{\sqrt{\pi n}}e^{-\frac{x^2}{n}}\approx 2^{-2n}\left(
    \begin{matrix}
        2n\\
        n+x
    \end{matrix}
    \right),~~~~\forall x=-2^{\ell-1}+1,\cdots,2^{\ell-1}-1,
\end{equation}
and $a_x=0$ for $x=-2^{\ell-1}$. We set
\begin{equation}
 n=\left\lfloor\frac{1}{2}\left\lfloor
       \left(\frac{2^{\ell-1}}{x_{\rm tune}}\right)^2
      \right\rfloor\right\rfloor,
\end{equation}
where $x_{\rm tune}$ is chosen to upper bound the error of truncating the sum to a sum up to a finite cutoff as in Eq.~\eqref{eq:sum} to be $O(e^{-x_{\rm tune}^2/2})$ for a Gaussian distribution with variance $n/2$.

The standard deviation of a CSS is $\sigma=\sqrt{N}/2$, where we work with the spin operators like $\hat{J}_z$ in this section which differs by a factor of $2$ than the previous $Z$. We choose the rotation angle to be
\begin{equation}
    4\pi\phi=\frac{2}{q_{\rm factor}\sigma},
\end{equation}
and for $a=1,\ldots,\ell$, we have
\begin{equation}
    \phi_a=2^{a-1}\pi\phi.
\end{equation}
With this parameter setting, our protocol effectively implements the filter
\begin{equation}
\begin{aligned}
    \hat{P}_{\sigma'}&=\sum_{m=-j}^{j}\exp{\left[-\frac{(m-m_0)^2}{2\sigma'^2}\right]}|j,m\rangle\langle j,m|\\
    &=\exp{\left[-\frac{(\hat{J}_z-m_0)^2}{2\sigma'^2}\right]}\\
    &\approx\cos{^{\left\lfloor\frac{\alpha^2}{\sigma'^2}\right\rfloor_{\mathrm{even}}}\left(\frac{\hat{J}_z-m_0}{\alpha}\right)},
\end{aligned}
\end{equation}
where $\alpha=q_{\rm f}\sigma$, $2n=\left\lfloor\frac{\alpha^2}{\sigma'^2}\right\rfloor_{\mathrm{even}}$, and $2^{\ell}\phi m_0$ is given by which QFT window picked after measuring the $\ell$ ancilla qubits.

The default setting is $x_{\rm tune}=1.0$ and $q_{\rm f}=1.5$.

\subsection{Noise parameter settings}

We focus on the effect of stochastic incoherent noise in the device. Coherent gate errors are not considered.

The noise level at noise scale $\lambda=1$ is set to be their typical level in the current superconducting device, as summarized in
Table~\ref{tab:noise}. By setting the noise scale to be $\lambda$, we multiply every error rate entry in
the column, including initialization, by $\lambda$.

\begin{table}[ht]
\centering
\small
\begin{tabular}{@{}lll@{}}
\hline\hline
Parameter & Value & Meaning in the simulator\\
\hline
$p_1$ & $2\times10^{-4}$ & total one-qubit depolarizing error rate\\
$p_2$ & $8\times10^{-3}$ & total 15-Pauli error rate after every \textsf{CNOT}\\
$p_{\rm meas}$ & $10^{-2}$ & classical bit flip rate of one qubit readout\\
$p_{\rm idle}$ & $10^{-5}$ & total one-qubit depolarizing error rate per idle tick\\
$p_{\rm init}$ & $10^{-3}$ & computational-basis reset error rate\\
\hline\hline
\end{tabular}
\caption{Noise rates at $\lambda=1$.}
\label{tab:noise}
\end{table}

In our simulation, we assume the \textsf{PREP}, $G$, $G^\dagger$, projection to $m=0$ QFT window, and final data offset rotation to be ideal, and we focus on the middle control unitaries $\mathcal{U}$ which dominate the errors. 

\subsection{Efficient exact branch representation}

In our simulation, we never store the full $2^{\ell+N}$ statevector or its compressed tensor network representation.  For one sampled trajectory $\tau$, we store each term in
\begin{equation}
 \ket{\Psi_\tau}=\sum_{x=-2^{\ell-1}}^{2^{\ell-1}-1}a_x\ket{x}
 \bigotimes_{i=1}^{N}\ket{\varphi^{(\tau)}_{x,i}},
 \label{eq:product-branch}
\end{equation}
which is just a product state of $\ell+N$ qubit. Conditioned on a computational-basis reference branch $|x\rangle$, each ideal controlled unitary and each sampled incoherent Pauli faults act as a product of one-qubit operations on the data qubits. So to simulate the protocol under the noises, we only need to sample the trajctories and keep track of each branch: $Z$ faults at reference qubit $a$ multiply $a_x$ by $-1$; $X$ faults at reference qubit $a$ permute $m\leftrightarrow m\oplus 2^{a-1}$.  Data Pauli frames are applied directly to the local spinors. Most importantly, an incoming data $X$ error is retained in the spinor before the
next non-Clifford $R_z$ rotation, and its effect
\begin{equation}
 R_z(\phi_a)X=XR_z(-\phi_a)
\end{equation}
is automatically handled. Similarly, an incoming data $Z$ error that is already present when the next controlled block is applied has the effect
\begin{equation}
 U'_a Z_i=Z_aZ_iU'_a,
\end{equation}
where $U'_a$ is defined as Eq.~\eqref{eq:Uaprime}.

\subsection{Binary Pauli frames and channel conventions}
A Pauli frame is stored as two bits
\begin{equation}
 P(x,z)=X^xZ^z,\qquad (x,z)\in\mathbb{F}_2^2.
\end{equation}
Thus $I=(0,0)$, $X=(1,0)$, $Y=(1,1)$ up to phase, and $Z=(0,1)$. Notice that global phase is irrelevant. Conjugating an existing Pauli frame through an ideal \textsf{CNOT} gives
\begin{equation}
 x_t\leftarrow x_t\oplus x_c,\qquad
 z_c\leftarrow z_c\oplus z_t.
 \label{eq:cnot-frame}
\end{equation}
The \textsf{CNOT} itself can induce errors, i.e.
\begin{equation}
\text{noisy}~\textsf{CNOT}=(E_c\otimes E_t)\textsf{CNOT}_{c\to t},    
\end{equation}
where $E_c\otimes E_t$ is a randomly sampled two-qubit Pauli error applied after the ideal CNOT. So after the propagation, we need to update the Pauli frame by XOR the two-qubit Pauli errors induced by the \textsf{CNOT}.

\subsubsection{One-qubit depolarizing channel}
The parameter $p_1$ is the total non-identity probability, not the probability of each Pauli, so
\begin{equation}
 \mathcal E_1(\rho)=(1-p_1)\rho+\frac{p_1}{3}\sum_{P\in\{X,Y,Z\}}P\rho P.
 \label{eq:p1}
\end{equation}
For the baseline $p_1=2\times10^{-4}$,
\begin{equation}
 P(X)=P(Y)=P(Z)=6.6667\times10^{-5},\qquad
 P(x=1)=P(z=1)=1.3333\times10^{-4}.
\end{equation}

\subsubsection{Two-qubit depolarizing channel}
The parameter $p_2$ is also the total non-identity probability:
\begin{equation}
 \mathcal E_2(\rho)=(1-p_2)\rho+
 \frac{p_2}{15}
 \sum_{(P,Q)\ne(I,I)}(P\otimes Q)\rho(P\otimes Q).
 \label{eq:p2}
\end{equation}
The sampler first propagates the pre-existing frame through the ideal \textsf{CNOT} and then, with probability $p_2$, chooses one of the 15 non-identity pairs uniformly. Therefore, for either one of the two qubits,
\begin{equation}
\begin{aligned}
 &P(X)=P(Y)=P(Z)=\frac{4p_2}{15},\\
 &P(P\ne I)=\frac{12p_2}{15}=\frac{4p_2}{5},\\
 &P(x=1)=P(z=1)=\frac{8p_2}{15}. 
\end{aligned}
 \label{eq:p2-marginals}   
\end{equation}
The two qubits are \emph{not} sampled independently, and the spatial correlations indicated in Eq.~\eqref{eq:p2} are retained. At $p_2=8\times10^{-3}$, each nonidentity pair has probability $5.3333\times10^{-4}$, each 1-qubit Pauli marginal is $2.1333\times10^{-3}$, and each 1-qubit nonidentity probability is $6.4\times10^{-3}$.

\subsection{Initialization settings}
A data qubit is initialized to $\ket 0$. With probability $p_{\rm init}$ it is instead $\ket 1=X\ket 0$. An ideal Hadamard maps it to $\ket +$ and maps the reset error $X$ to a $Z$ frame fault. A full one-qubit depolarizing channel of total strength $p_1$ follows the Hadamard. Since $X\ket +=\ket +$, only the $Z$-frame component in the depolarizing channel matters at this point. The sampled $Z$ frame is immediately applied to the exact product branches. 

Every controlled unitary block uses fresh edge qubits. An edge qubit is reset to $\ket 0$, and with probability $p_{\rm init}$ it has an $X$ frame. Each star $S_v$ measurement round uses a fresh ancilla. Its wrong-$\ket 0$ reset becomes a $Z$ fault after the ideal Hadamard, and a total-$p_1$ one-qubit depolarizing channel follows that Hadamard. Again, only the resulting $Z$ component affects a freshly prepared $\ket +$ ancilla before its first \textsf{CNOT}.

\subsection{One noisy controlled unitary block}
We apply the ideal net operation on the branch representation of the reference and data qubits analytically, without explicitly storing the edge qubits and star $S_v$ measurement ancilla qubits in a statevector, and separately sample the residual errors relative to the ideal block.

For block $a$, the ideal net operation is
\begin{equation}\label{eq:Uaprime}
U'_a=
\left(
|0\rangle_a\langle0|\otimes I
+
|1\rangle_a\langle1|\otimes X^{\otimes N}
\right)
\left(
I\otimes R_z(\phi_a)^{\otimes N}
\right).    
\end{equation}
It contains the exact data rotation
$$
   R_z(\phi_a)^{\otimes N},
$$
and the exact quantum fanout gate controlled by reference qubit $a$
$$
   |0\rangle_a\langle0|\otimes I+|1\rangle_a\langle1|\otimes X^{\otimes N}.
$$
The second operation is the net noiseless effect of the reference-to-edge \textsf{CNOT}'s, star $S_v$ measurements, edge measurements, and ideal feed-forward. 

We then separately run the physical circuit at the Pauli-frame level to determine the deviation from this ideal operation. We sample faults from the data $R_z$ locations, reference-to-edge \textsf{CNOT}'s, two rounds star $S_v$ measurement circuits, ancilla and edge readouts, idles, and the reference and data qubit feed-forward operations. After decoding, this circuit-level sampler returns an incremental residual Pauli frame
$$
P_\tau
=
P_a^{(\tau)}
\otimes
\bigotimes_{i=1}^N P_i^{(\tau)}.
$$
The simulated block is therefore represented as $P_\tau U'_a$,
with the correct timing of the faults accounted for.

Each trajectory returns the Pauli frame on the active reference and every data vertex. We then apply these Paulis to Eq.~\eqref{eq:product-branch}. The next block receives the actual faulty product branches, while its new sampler starts from a zero frame.

\subsection{MWPM decoder and edge $X$ corrections}

The MWPM decoder aims to detect and correct the edge qubit errors such that one can perform feedforward operations correctly on the data qubits according to clean edge records to realize the quantum fanout gate (excluding the phase part). 

Let a vertex-domain bit $a_v\in\{0,1\}$ denote if $I$ or $X$ is applied to a data qubit in the ideal case. The ideal edge record is
$$
d_{uv}=a_u\oplus a_v.
$$
Let $r_{uv}$ denote the residual edge error, so the measured edge record is
$$
d'_{uv}=d_{uv}\oplus r_{uv}.
$$
Introduce Ising variables
$$
\sigma_v=(-1)^{a_v},
\qquad
s_{uv}=(-1)^{d'_{uv}}.
$$
For a proposed domain configuration $\{\sigma_v\}$,
$$
(-1)^{r_{uv}}=(-1)^{d_{uv}\oplus d'_{uv}}=s_{uv}\sigma_u\sigma_v.
$$
Therefore $s_{uv}\sigma_u\sigma_v=+1$ means the proposed domain configuration agrees with the measured edge record and $s_{uv}\sigma_u\sigma_v=-1$ means an edge error happened. 

If we assume that the residual edge errors are \emph{independent}, with error rate $P(r_{uv}=1)=q_{e}$ with $e=(u,v)$. Notice that the edges at the boundary of the open square lattice experience different numbers of gates and thus should have a different error rate from the edges in the bulk. Then
$$
\begin{aligned}
P(\bm{\sigma}\mid \bm{s}) &\propto\prod_{\langle u,v\rangle} (1-q_e)^{(1+s_{uv}\sigma_u\sigma_v)/2}q_e^{(1-s_{uv}\sigma_u\sigma_v)/2}\\
&\propto\exp\left(
\sum_{\langle u,v\rangle} K_es_{uv}\sigma_u\sigma_v
\right),
\end{aligned}
$$
where
$$
K_e=\frac12\log\frac{1-q_e}{q_e}.
$$
This is a random-bond Ising model (RBIM)
$$
H=-\sum_{\langle u,v\rangle} J_es_{uv}\sigma_u\sigma_v
$$
at the Nishimori line which satisfies
$$
e^{-2\beta J_e}=\frac{q_e}{1-q_e}
$$
and $\beta\rightarrow\infty$ when $q_e\ll 1$. The measured edge record $d'_{uv}$ provide the random signs $s_{uv}$, while the domain variables $a_v$ become the Ising spins.

However, our circuit does not generate perfectly independent residual edge errors, because a star-ancilla fault can propagate to several edges, two-qubit \textsf{CNOT} faults correlate their endpoints, and edge faults can be correlated with data and reference faults. Consequently, the exact posterior is not a simple nearest-neighbor independent-bond RBIM. It is a correlated-disorder model containing multibond interactions. An exact Bayesian decoder using the complete circuit-level likelihood would lie on a generalized Nishimori surface. Here we use a simple MWPM decoder which discards those correlations and approximates the problem by an independent-bond, \emph{zero}-temperature RBIM.

If each edge record is correct, we will have $\prod_{e\in \partial p}s_e=+1$ for every plaquette $p$. As long as there are edge $X$ errors, $\prod_{e\in \partial p}s_e=-1$ reveals the endpoints of the error chains. MWPM estimates the error pattern by minimizing the weight of the antiferromagnetic bonds among all the configurations which is gauge equivalent to the measured edge record $\bm{d}'$, i.e.
$$
\hat{\bm{r}}
=
\arg\min_{\bm{\tilde{d}}\sim\bm{d}'}
\sum_e w_e\tilde{d}_e,
$$
with
$$
w_e=\log\frac{1-q_e}{q_e}=2K_e.
$$
Then we can repair the edge record
$$
\hat{\bm{d}}=\bm{d}'\oplus\hat{\bm{r}}.
$$
There can be different possible $\hat{\bm{r}}$ that satisfies the same endpoint constraint. Due to this degeneracy, the error chain chosen by the MWPM is not necessarily the true one. This is not a problem the decoding in toric code, but can leave residual domains in our task, as can be seen from the following.

On an open square lattice, we then find $\hat{a}_v$ satisfying
$$
\hat{d}_{uv}=\hat{a}_u\oplus\hat{a}_v.
$$
We fix $\hat{a}_{\rm ref}=0$ (or equivalently, we can accumulate the corrected edge bits along a path from the reference qubit), and then apply $X^{\hat{a}_v}$ for each data vertex qubit $v$. The remaining domain pattern is
$$
\delta_v=a_v\oplus\hat{a}_v,
$$
and its corresponding domain wall variable satisfies
$$
\delta_u\oplus\delta_v=r_{uv}\oplus\hat{r}_{uv}.
$$
Consequently, if $\hat{\bm{r}}=\bm{r}$, the domains are removed and we obtain the ideal quantum fanout gate (up to relative phase corrections); if $\hat{\bm{r}}\neq \bm{r}$, the discrepancy can leave residual domains, even though the repaired classical record passes every plaquette check.

Because we use open boundary condition for the 2d square lattice, the number of gates applied to the qubits at the boundary is different from the bulk. So we benchmark $w_e$ in MWPM to improve the performance of the decoder.

\subsection{Supplementary numerical results}
We optimize the rotation angle $\phi$, the Gaussian width $\sigma$, and the number of rounds $k$ for protocol in Fig.~\ref{fig:recursivebinary} or the number of ancillae $\ell$ for protocol in Fig.~\ref{fig:sss_circuit}. The results are shown below.
\begin{figure}[h!]
    \centering
    \includegraphics[width=0.497\linewidth]{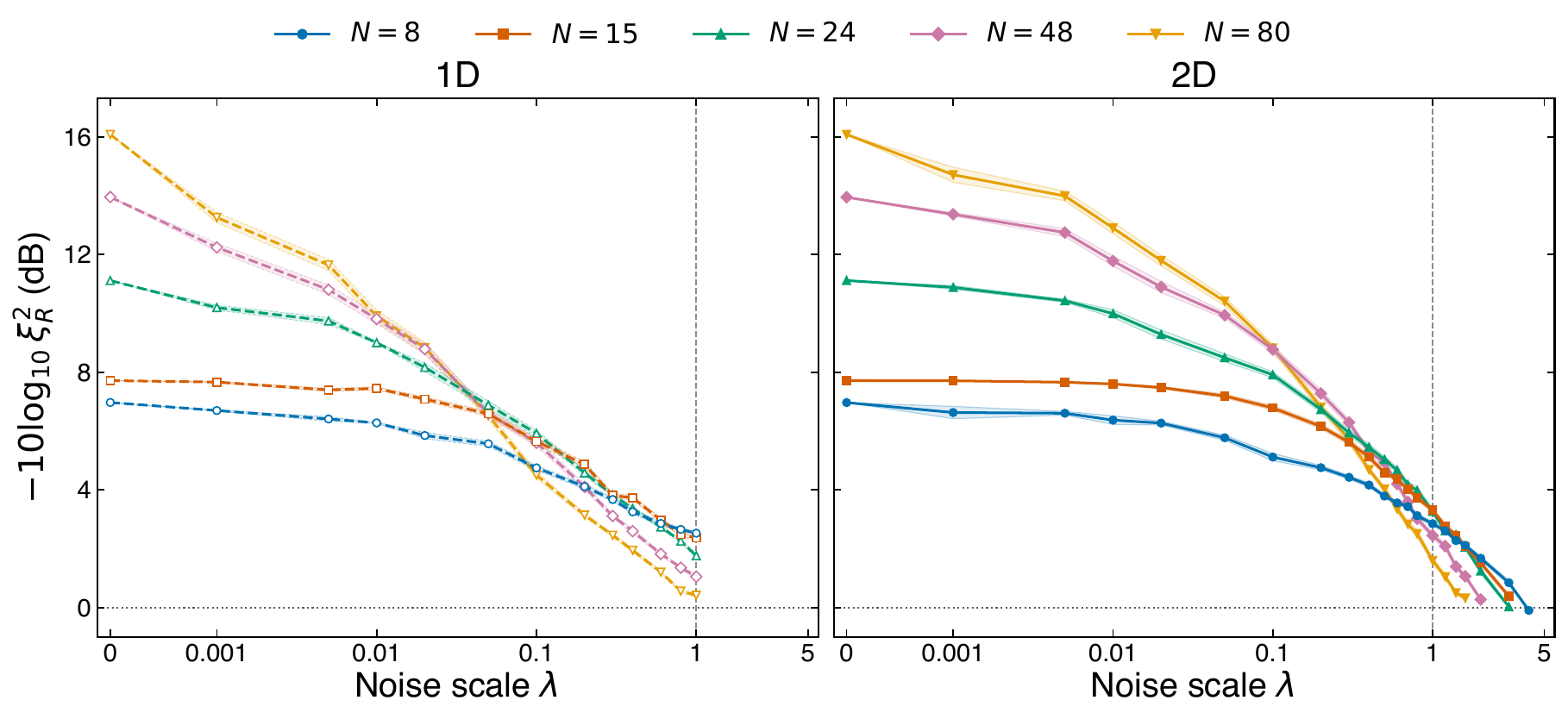}
    \includegraphics[width=0.497\linewidth]{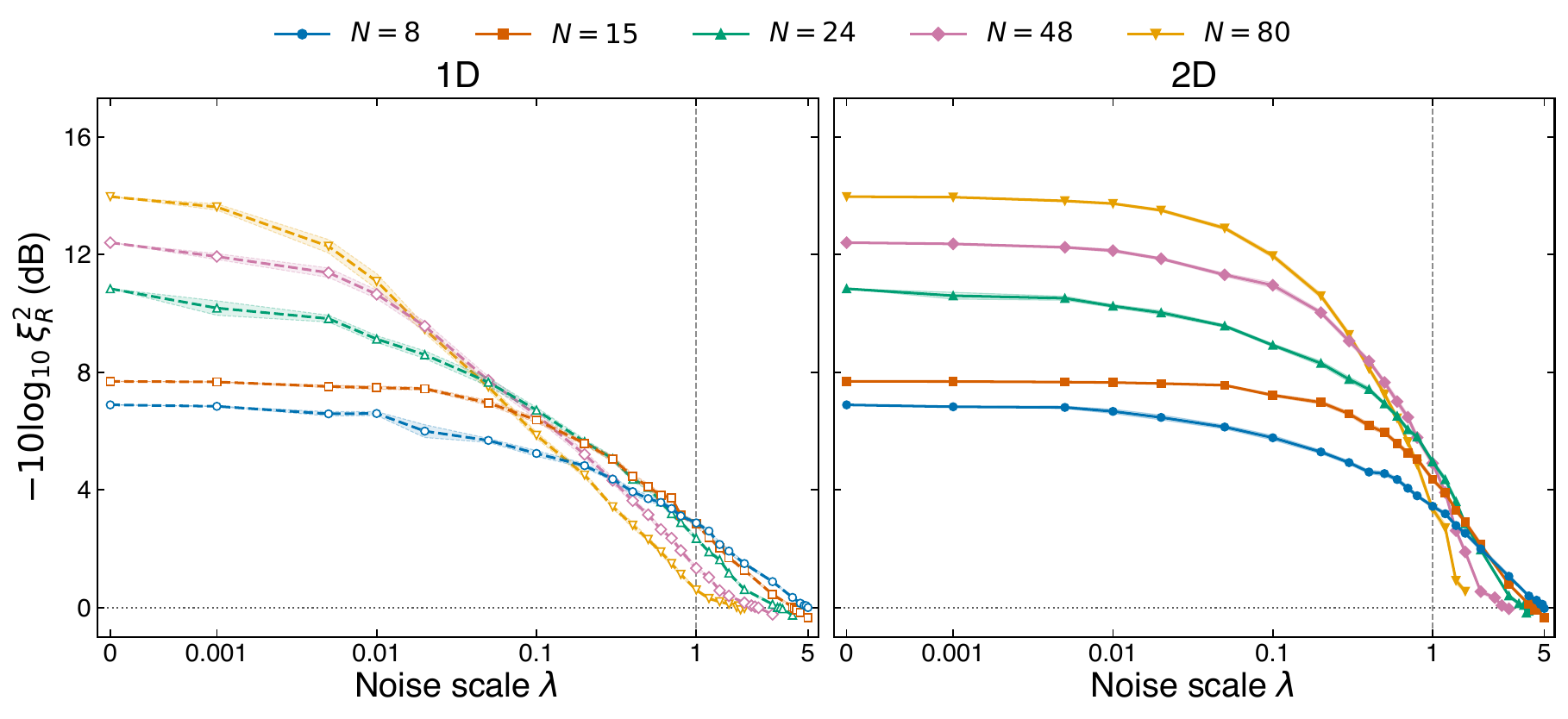}
    \caption{Metrological gain for the protocol in Fig.~\ref{fig:sss_circuit} (left) and in Fig.~\ref{fig:recursivebinary} (right).}
\end{figure}

\begin{figure}[h!]
    \centering
    \includegraphics[width=0.497\linewidth]{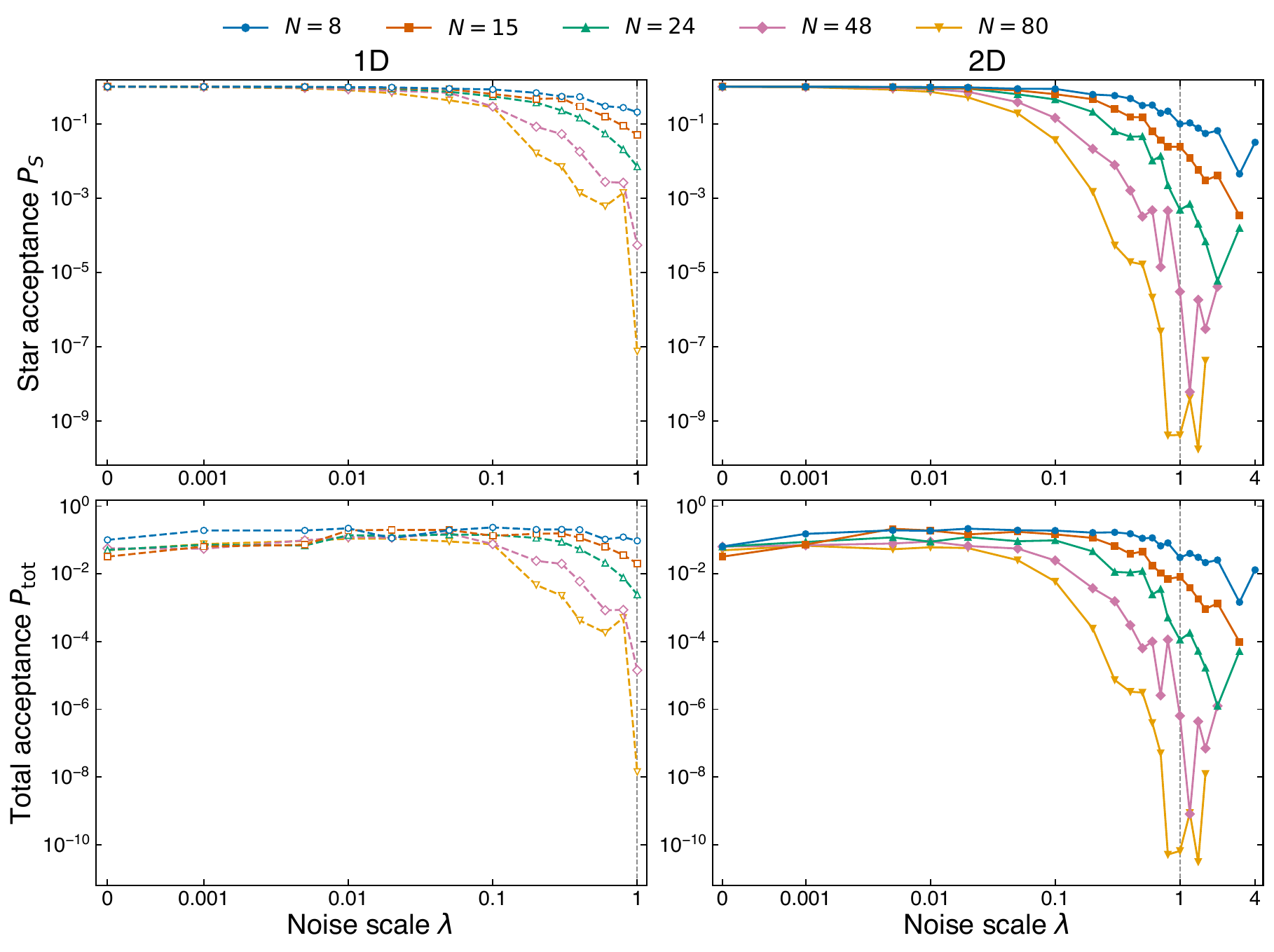}
    \includegraphics[width=0.497\linewidth]{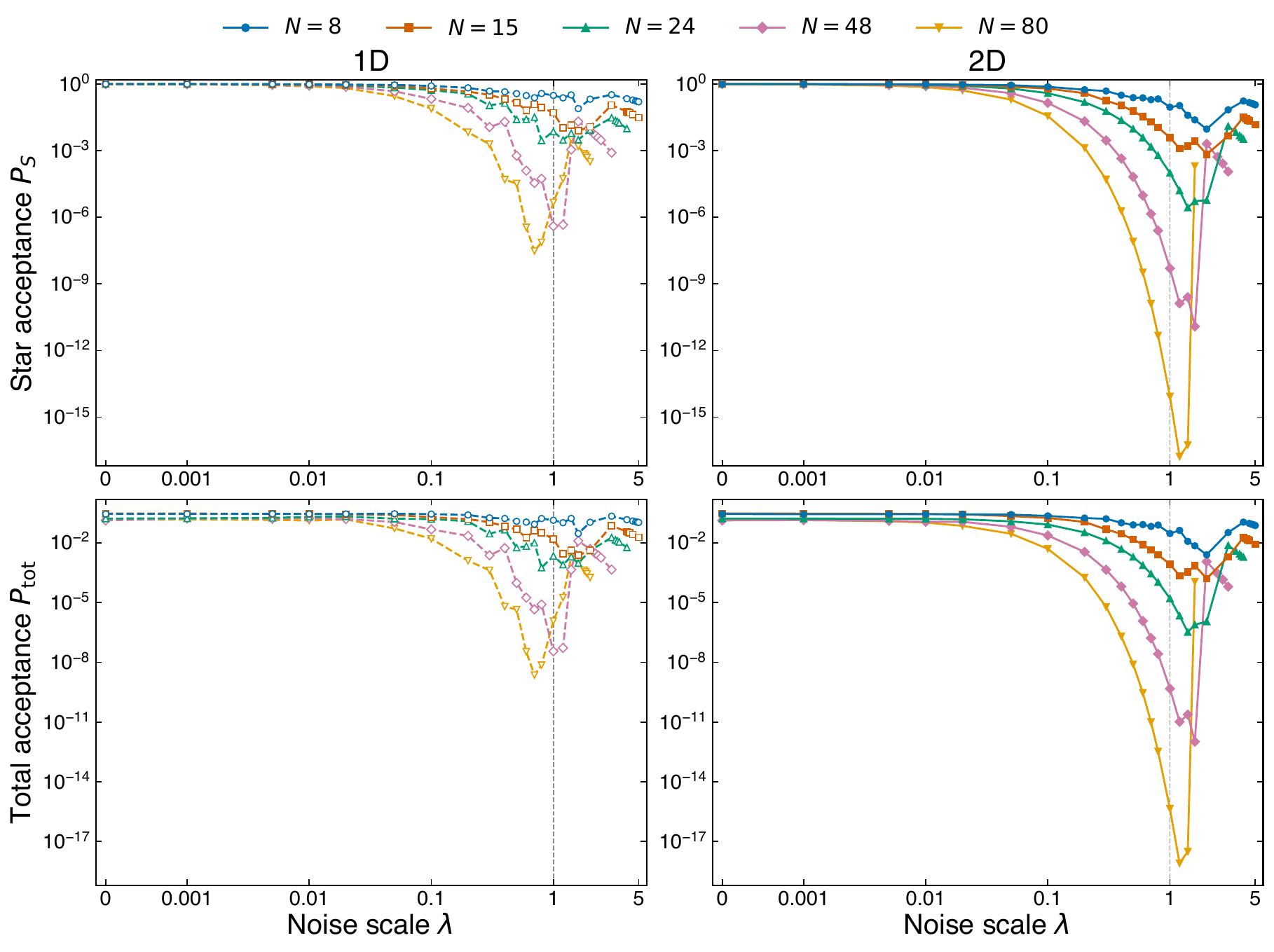}
    \caption{Acceptance rate for the protocol in Fig.~\ref{fig:sss_circuit} (left) and in Fig.~\ref{fig:recursivebinary} (right). The acceptance rates sometimes increases with larger $\lambda$ because $\ell$ or $k$ is chosen to be smaller.}
\end{figure}

\begin{figure}[h!]
    \centering
    \includegraphics[width=0.497\linewidth]{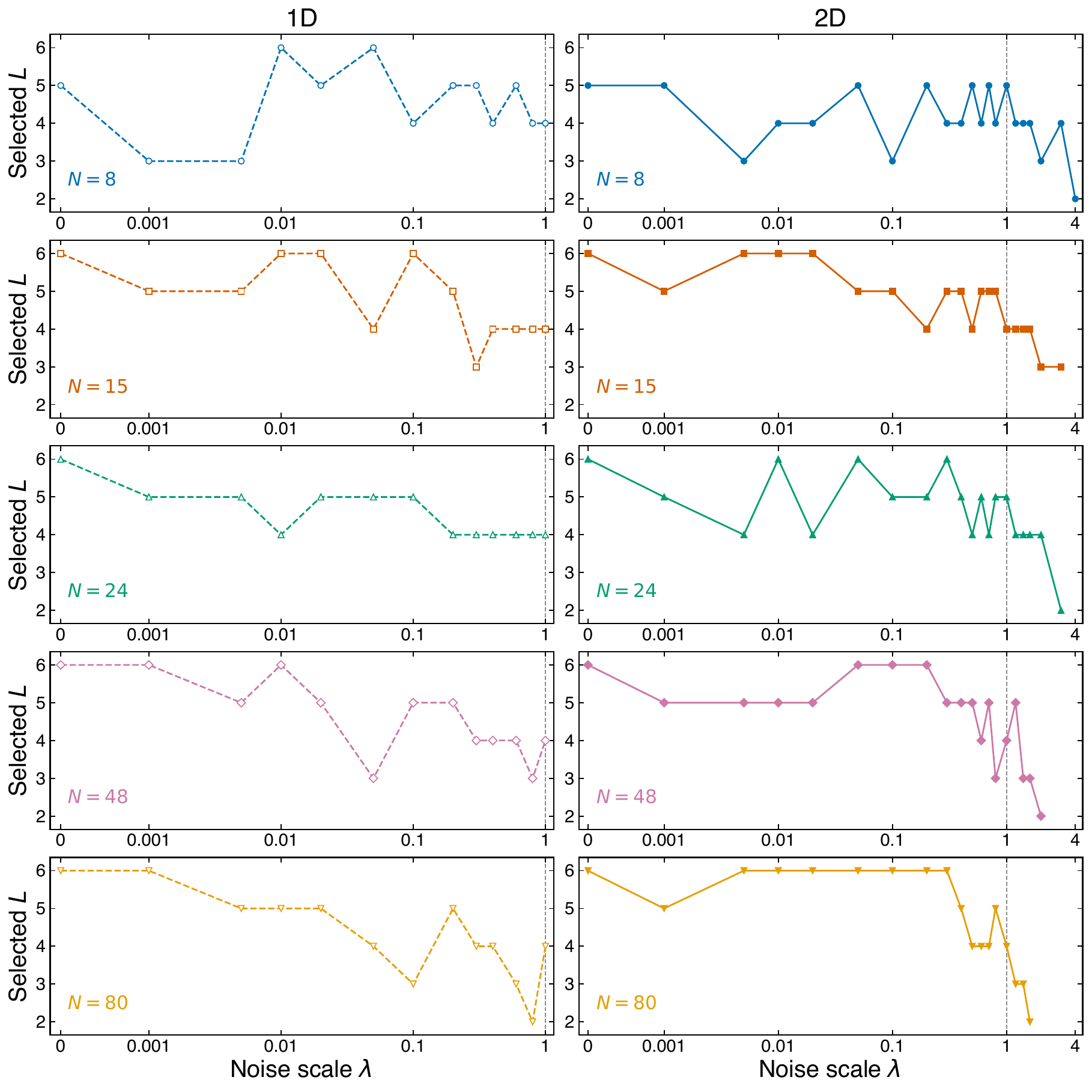}
    \includegraphics[width=0.497\linewidth]{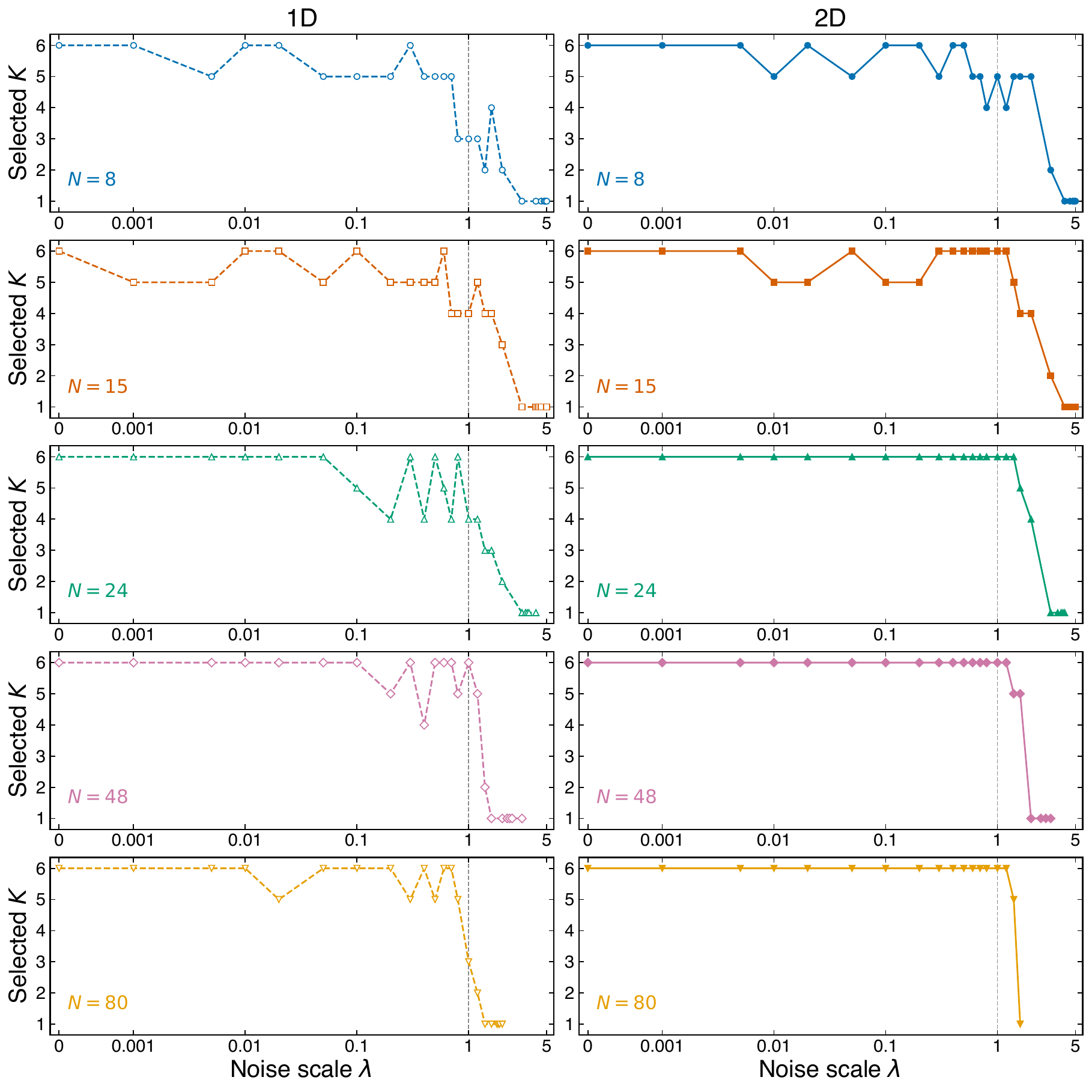}
    \caption{$\ell$ used for the protocol in Fig.~\ref{fig:sss_circuit} (left) and $k$ used for the protocol in Fig.~\ref{fig:recursivebinary} (right).}
\end{figure}

\begin{figure}[h!]
    \centering
    \includegraphics[width=0.58\linewidth]{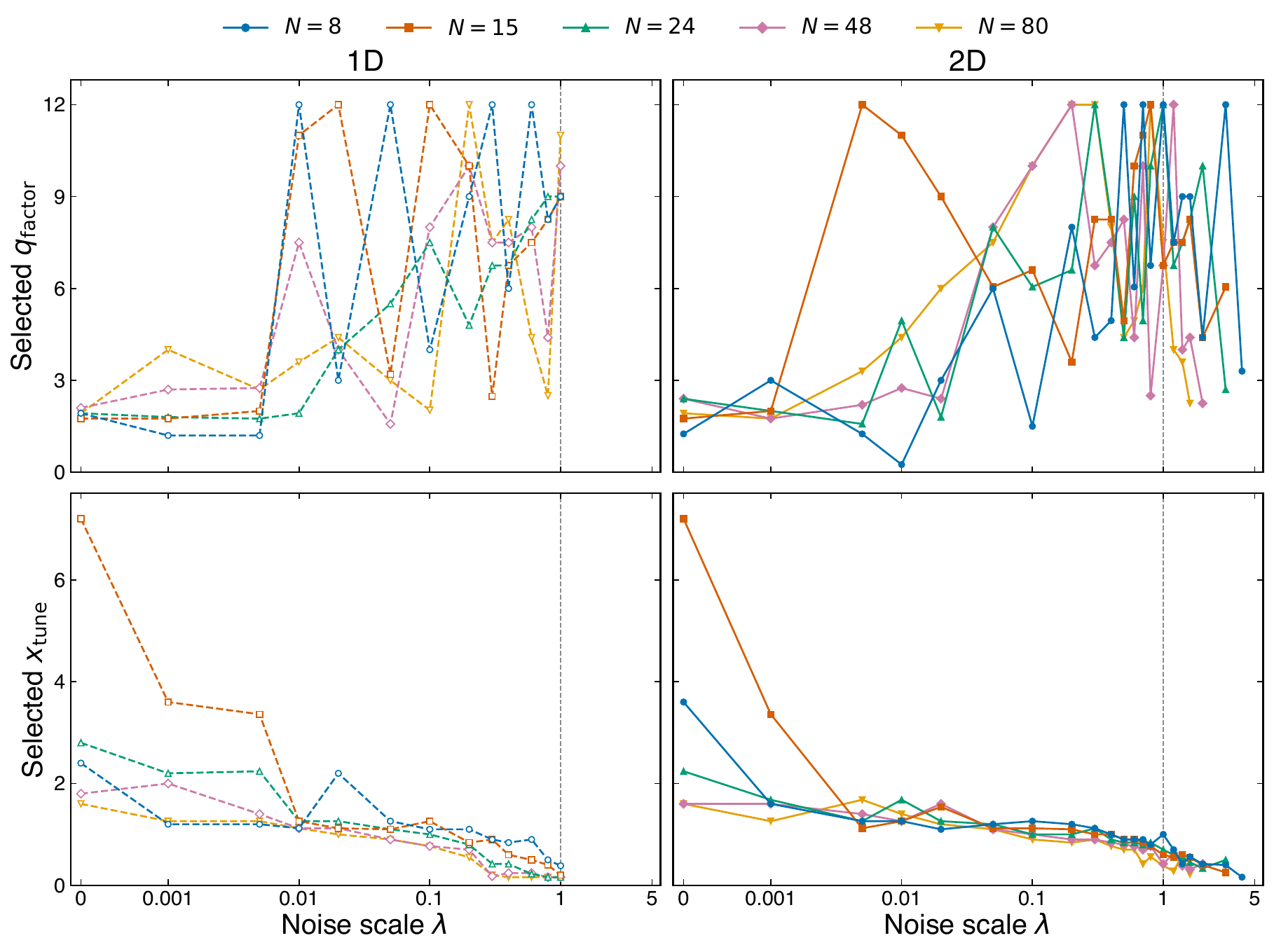}
    \includegraphics[width=0.58\linewidth]{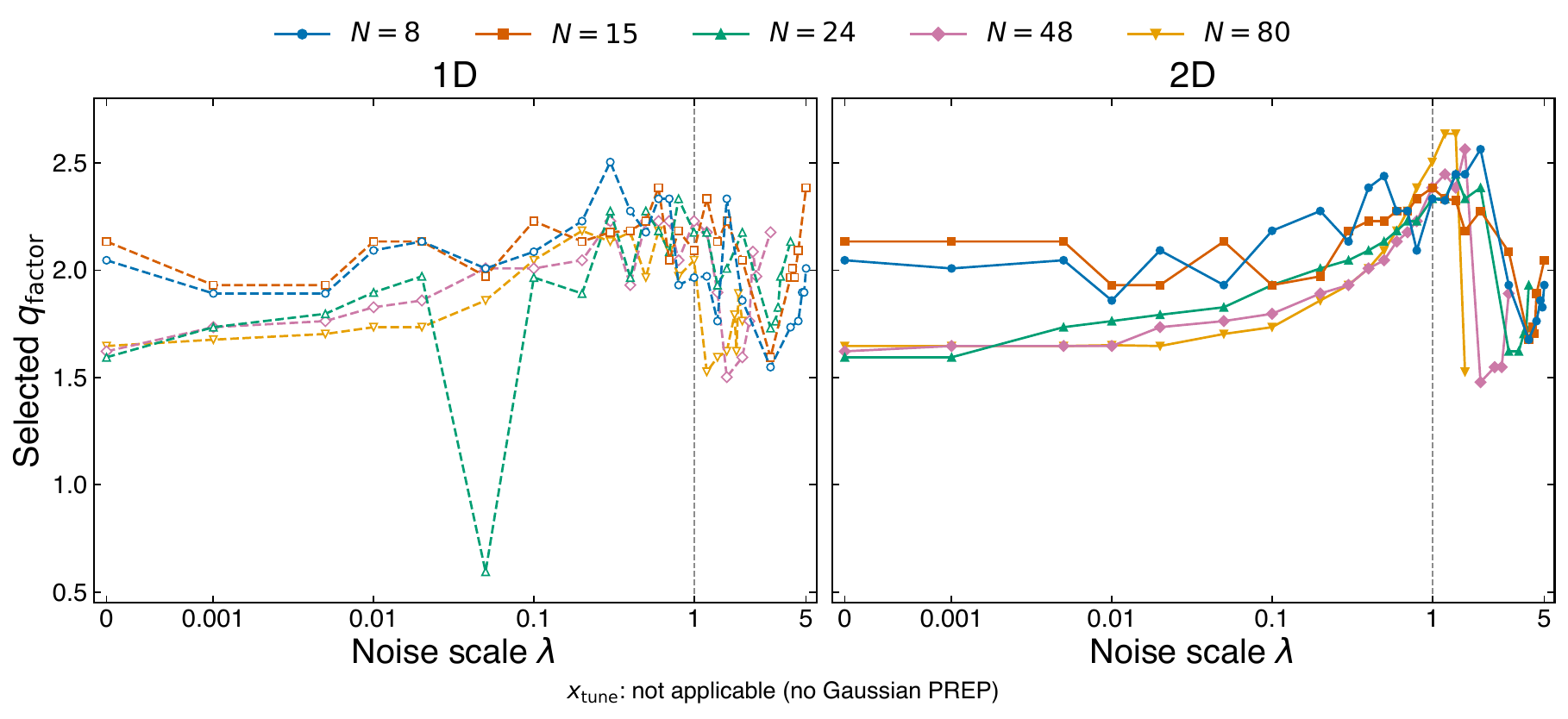}
    \caption{$q_{\rm factor}$ and $x_{\rm tune}$ used for the protocol in Fig.~\ref{fig:sss_circuit} (top) and $q_{\rm f}$ used for the protocol in Fig.~\ref{fig:recursivebinary} (bottom).}
\end{figure}

\begin{figure}
    \centering
    \includegraphics[width=0.6\linewidth]{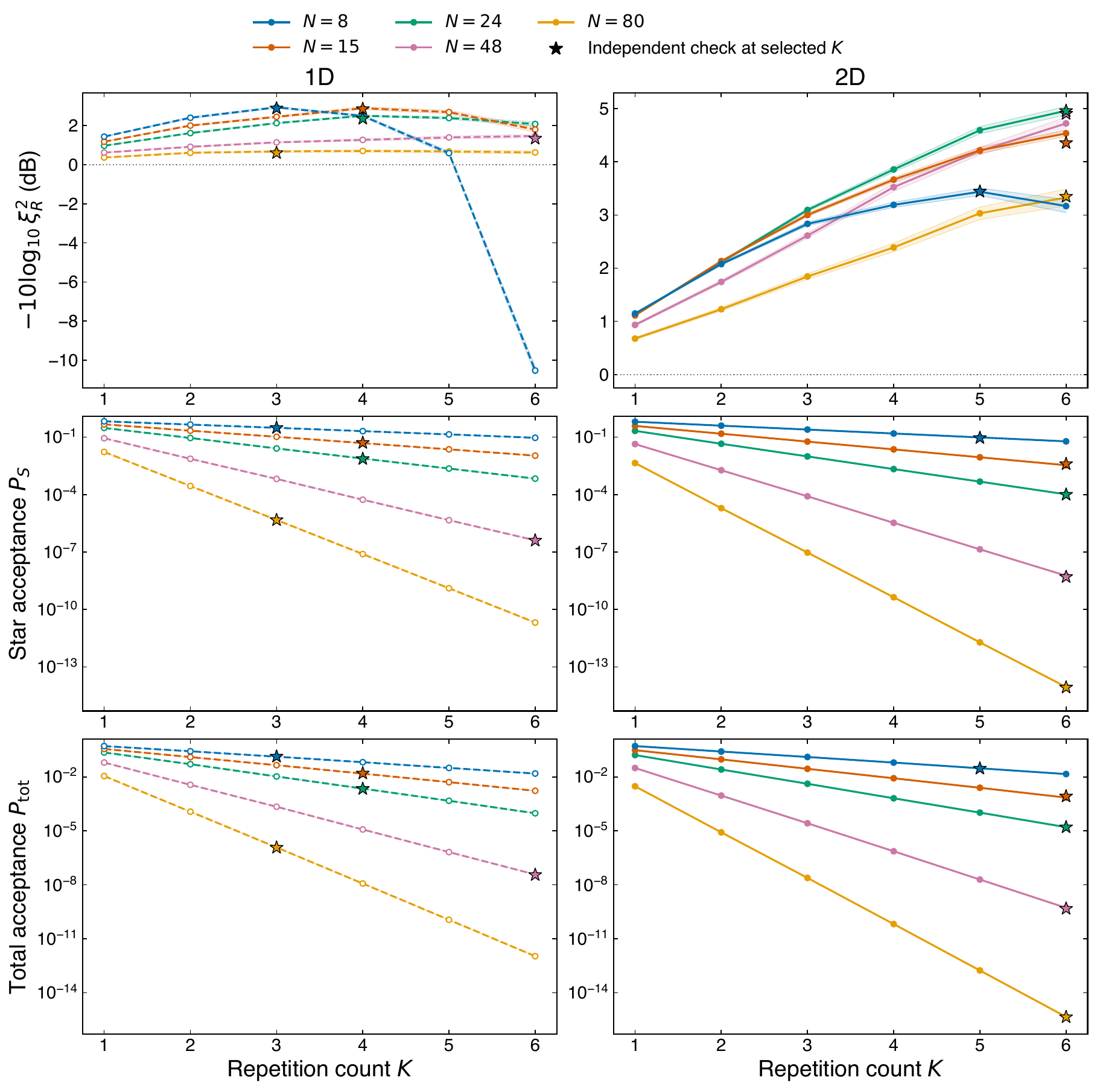}
    \caption{How the metrological gain the the acceptance rate change with the number of rounds $K$ for the protocol in Fig.~\ref{fig:recursivebinary} at $\lambda=1$.}
\end{figure}

\begin{figure}
    \centering
    \includegraphics[width=0.8\linewidth]{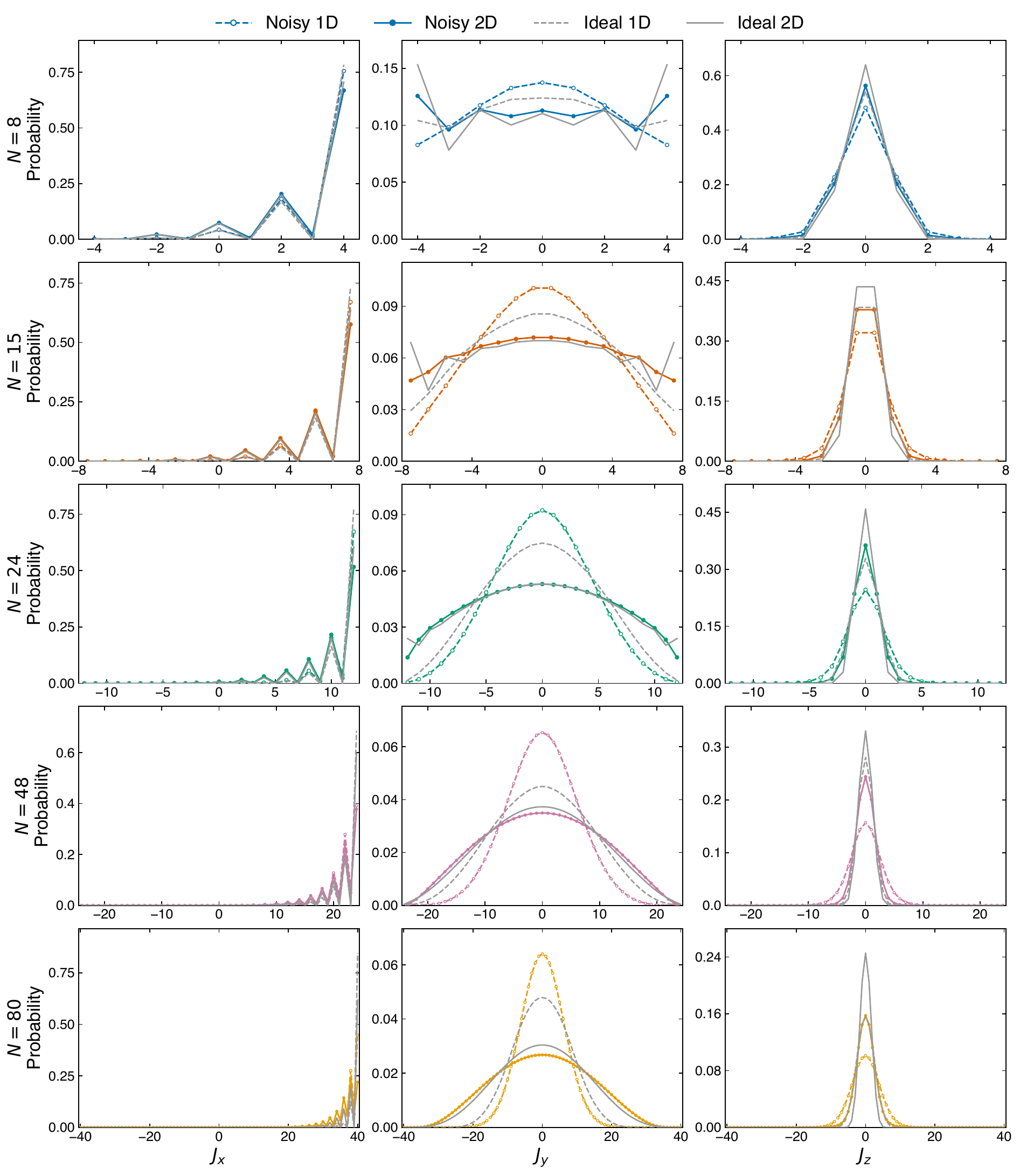}
    \caption{$J_x$, $J_y$, $J_z$ distribution of the output state for the protocol in Fig.~\ref{fig:recursivebinary} at $\lambda=1$ as compared to the ideal optimal.}
\end{figure}

\end{document}